\documentclass[]{interact}
\usepackage{environ} 
\usepackage[dvipsnames]{xcolor}
\usepackage{hyperref}
\usepackage{graphicx}
\usepackage[normalem]{ulem}
\graphicspath{{./gfx/}}

\usepackage{epstopdf}

\usepackage[font=small,labelfont=bf]{caption}
\usepackage{subcaption}
\DeclareCaptionFormat{sqa}{#1~#3}
\usepackage{natbib}
\bibpunct[, ]{(}{)}{;}{a}{}{,}

\usepackage{tikz}
\usetikzlibrary{arrows.meta}
\usetikzlibrary{calc}
\usetikzlibrary{patterns}

\NewEnviron{notforprint*}{\BODY}
\NewEnviron{trivial}{\BODY}
\NewEnviron{ignore*}{}
\NewEnviron{skipproof}{}
\NewEnviron{noignore*}{\BODY}
\newcommand{\ignore}[1]{}
\newcommand{\nolabel}[1]{}

\hypersetup{
    colorlinks=true,
    bookmarksnumbered=true,
    bookmarksopen=true,
    linkcolor=blue,
    citecolor=blue,
    urlcolor=blue,
    unicode=true,
    breaklinks=true
}
\renewcommand{\geq}{\geqslant}
\renewcommand{\leq}{\leqslant}

\tikzset{>={Latex[length=8pt, width=4pt]}}

\newcommand{\ThemeRed}{BrickRed}

\newcommand{\class}{\mathbb{C}}
\newcommand{\Durations}{\mathcal{D}}
\newcommand{\set}[1]{\left\{#1\right\}}
\newcommand{\smallset}[1]{\{#1\}}

\newcommand{\Bigset}[1]{\Big\{#1\Big\}}
\newcommand{\cset}[2]{\left\{#1 \,:\, #2\right\}}
\newcommand{\smallcset}[2]{\{#1 \,:\, #2\}}
\newcommand{\bigcset}[2]{\big\{#1 \,:\, #2\big\}}
\newcommand{\Bigcset}[2]{\Big\{#1 \,:\, #2\Big\}}

\newcommand{\slfrac}[3][\big]{#2#1/#3}
\newcommand{\ceiling}[1]{\left\lceil #1 \right\rceil}
\newcommand{\floor}[1]{\left\lfloor #1 \right\rfloor}
\let\Pr\relax
\DeclareMathOperator{\Pr}{\mathsf{P}} 
\DeclareMathOperator{\EV}{\mathsf{E}} 
\DeclareMathOperator{\Var}{\mathsf{Var}} 
\DeclareMathOperator{\Cov}{\mathsf{Cov}} 
\DeclareMathOperator{\SE}{\mathsf{SE}} 
\DeclareMathOperator{\LCPFA}{\mathsf{LPFA}} 
\DeclareMathOperator{\LUPFA}{\mathsf{LPFA}^\star} 
\DeclareMathOperator{\LPFA}{\mathsf{LPFA}}   
\DeclareMathOperator{\LPD}{\mathsf{LPD}}   
\DeclareMathOperator{\ARL}{\mathsf{ARL}}   
\DeclareMathOperator*{\esssup}{ess\,sup}
\DeclareMathOperator*{\essinf}{ess\,inf}

\DeclareMathOperator{\One}{\mathchoice{\rm 1\mskip-4.2mu l}{\rm 1\mskip-4.2mu l}{\rm 1\mskip-4.6mu l}{\rm 1\mskip-5.2mu l}}

\newcommand{\CUSUM}{\textsf{CS}}

\newcommand{\WLCUSUM}{\textsf{WL}}
\newcommand{\FMA}{\textsf{FMA}}
\newcommand{\MFMA}{\textsf{mFMA}}

\newcommand{\cF}{\mathcal{F}}

\newcommand{\cN}{\mathcal{N}}

\newcommand{\wtV}{\widetilde{V}}
\newcommand{\wtT}{\widetilde{T}}

\newcommand{\bY}{\mathbf{Y}}

\theoremstyle{plain}
\newtheorem{lemma}{Lemma}
\newtheorem*{lemma*}{Lemma}

\newtheorem*{theorem*}{Theorem}

\newtheorem*{corollary*}{Corollary}
\newtheorem{proposition}{Proposition}
\newtheorem*{proposition*}{Proposition}
\theoremstyle{remark}

\newtheorem*{assumption*}{Assumption}
\theoremstyle{definition}

\newtheorem*{definition*}{Definition}

\DeclareFontFamily{U}{matha}{\hyphenchar\font45}
\DeclareFontShape{U}{matha}{m}{n}{
      <5> <6> <7> <8> <9> <10> gen * matha
      <10.95> matha10 <12> <14.4> <17.28> <20.74> <24.88> matha12
      }{}
\DeclareSymbolFont{matha}{U}{matha}{m}{n}
\DeclareMathSymbol{\abscont}{3}{matha}{"21}

\begin{document}

\title{Detecting an Intermittent Change of Unknown Duration}

\author{
    \name{Grigory Sokolov\textsuperscript{a}, Valentin~S. Spivak\textsuperscript{b}, and Alexander~G. Tartakovsky\textsuperscript{c}}
    \affil{%
        \textsuperscript{a}Xavier University, Cincinnati, OH, USA; %
        \textsuperscript{b}Space Informatics Laboratory, Moscow Institute of Physics and Technology, Dolgoprudny, Moscow Region, Russia; %
        \textsuperscript{c}AGT StatConsult, Los Angeles, California, USA}
}

\maketitle

\begin{abstract}
Oftentimes in practice, the observed process changes statistical properties at an unknown point in time and the duration of a change is substantially finite, in which case one says that the change is intermittent or transient.
We provide an overview of existing approaches for intermittent change detection and advocate in favor of a particular setting driven by the intermittent nature of the change. We propose a novel optimization criterion that is more appropriate for many applied areas such as the detection of threats in physical-computer systems, near-Earth space informatics, epidemiology, pharmacokinetics, etc.
We argue that controlling the local conditional probability of a false alarm, rather than the familiar average run length to a false alarm, and maximizing the local conditional probability of detection is a more reasonable approach versus a traditional quickest change detection approach that requires minimizing the expected delay to detection.
We adopt the maximum likelihood (ML) approach with respect to the change duration and show that several commonly used detection rules (CUSUM, window-limited CUSUM, and FMA) are equivalent to the ML-based stopping times. We discuss how to choose design parameters for these rules and provide a comprehensive simulation study to corroborate intuitive expectations.
\end{abstract}


\section{Introduction} \label{sec:intro}

The problem of detecting intermittent (or transient) changes is motivated by a variety of applications such as
aerospace navigation and flight systems integrity monitoring~\cite[Ch 11]{TNB_book2014},
cyber-security~\cite{Debaretal-CN99,EllisSpeed-Book01,Kent,peng-lncs04,Tartakovskyetal-SM06,Tartakovskyetal-IEEESP06,Tartakovsky-Cybersecurity14},
identification of terrorist activity~\cite{Raghavanetal-AoAS2013},
industrial monitoring~\cite{duncan-book86,Jeskeetal-ASMBI2018},
air pollution monitoring~\cite{NikiforovIFAC2022},
radar, sonar, and electrooptics surveillance systems~\cite{BarshalomLi93,Black1,Jeskeetal-WileyRef2018,Tartakovsky-IEEEASC2002,Tartakovsky&Brown-IEEEAES08}, \cite[Ch 8]{Tartakovsky_book2020}.
As a result, it has been of interest to many practitioners for some time.
However, in contrast to classical quickest change-point detection, where one's aim is to detect persistent changes in the distribution of the signal as quickly as possible, minimizing the expected delay to detection assuming the change is in effect (see, e.g., \cite{TNB_book2014,Tartakovsky_book2020} and references therein) a lot fewer publications are devoted to the sequential detection of intermittent changes \cite{Wang+Willett:2005a,Wang+Willett:2005b,Ortner+Nehorai:2007,Nikiforov+et+al:2012,Moustakides:2014,Nikiforov+et+al:2017,Rovatsos+Zou+Veeravalli:2017,Nikiforov+et+al:2023}.

As opposed to the classical change-point detection problem, transient changes only last for a finite (often short) time. Several main scenarios motivate intermittent change detection.
In one of them, the under-change mode lasts for a finite and unknown (possibly random) time and can be detected with a certain delay even after it ends. Examples of such scenarios may be air pollution, water contamination, and pharmacokinetics/dynamics. In such a case, the standard approach to detecting changes that prescribes minimizing an average detection delay with a given false alarm rate may be appropriate.
Another scenario is when a change is associated with a critical anomaly like a threat (a terrorist physical or computer attack, military targets/threats, etc.) that appears and disappears at unknown points in time and should be detected not only as soon as possible but with a delay not bigger than a prescribed value.
In this second case, if the change is detected with a delay larger than a given time-to-alert, then it is considered completely missed.
For example, if a threat hits an object, then it is too late to detect it.
In this case, a conventional quickest change detection problem setup is not appropriate; rather, one aims to maximize the probability of detection within a prescribed time (or space) interval for a given false alarm rate.
There is also a third scenario where the change may last very long but still it may be required to detect it in a fixed relatively short interval while maximizing the probability of detection.
An example is target track initiation which should be performed in a very short interval, while tracks last a very long time and should be estimated with maximal possible accuracy (see, e.g., \cite{SpivakTar_EnT2020}).
Therefore, in some cases, taking the quickest detection approach (\cite{Ebrahimzadeh:2015,Zou+Fellouris+Veeravalli:SIT:2017, Zou+Fellouris+Veeravalli:TIT:2017, Rovatsos+Zou+Veeravalli:2017}) may seem reasonable;
in other cases, it can be argued (\cite{Broder+Schwartz:1990,Mei-SQA08, Tartakovsky-SQA08a,Bakhache+Nikiforov:2000,TNB_book2014,Tartakovsky_book2020}) that a substantially different criterion is called for.
Specifically, one's objective would be to develop and study detection rules that locally maximize the probability of detection (rather than its speed) subject to a constraint on the local probability of false alarm. As in \cite{Tartakovsky_book2020,Tartakovsky+et+al:2021}, we call this approach \emph{reliable} change-point detection.

In this work, we consider the case that the change duration is unknown (focusing on the case that it is deterministic) and examine how the maximum likelihood ratio approach for various degrees of information on the duration of the intermittent change yields some well-known change detection rules.
Specifically, we consider three detection procedures: Cumulative Sum (CUSUM) procedure, window-limited CUSUM, and finite moving average (FMA) procedure, and compare their performance in terms of the probability of detection versus the probability of false alarm.

This paper continues the research on reliable change detection along the lines started by \cite{Nikiforov+et+al:2017,Tartakovsky_book2020, BerenkovTarKol_EnT2020, Nikiforov+et+al:2023} and its contribution is as follows:

\begin{enumerate}
    \item We propose a new optimality criterion for reliable change detection of intermittent changes that maximizes the local conditional probability of detection for a given local conditional probability of false alarm (PFA).

    \item We show that the conditional PFA criterion is the most stringent compared to other popular ones such as controlling average run length (ARL) to false alarm or unconditional PFA.

    \item We use the maximum likelihood (ML) principle in the context of intermittent change detection to derive three popular rules, propose a modification of the FMA detection algorithm driven by its ML origins and show that it has better operating characteristics compared to the standard one.

    \item We propose ways to design the aforementioned rules. In particular, for CUSUM we adapt the integral equations framework and develop efficient numerical methods that allow for an almost precise evaluation of its operating characteristics. For window-limited rules, we provide the means to control (upper-bound) their local conditional PFA and perform a simulation study to examine its accuracy.

    \item We validate our analysis and compare CUSUM, window-limited CUSUM, and two FMA algorithms through numerical results.
\end{enumerate}

The rest of the paper is organized as follows.
In Section~\ref{sec:formulations}, we discuss various formulations of transient change detection problems and how they relate to each other. We propose one particular formulation and lay the common ground for the remainder of the paper.
In Section~\ref{sec:candidates}, we take the maximum likelihood approach to derive three rules commonly seen in literature in the context of detecting transient changes: CUSUM, window-limited CUSUM, and FMA. We further propose a modification of FMA driven by the maximum-likelihood nature of this rule.
In Section~\ref{sec:procedure_design}, we propose methods for the choice of design parameters for the rules under investigation, each driven by the specific nature of its detection statistic.
In Section~\ref{sec:simulations}, we perform a numerical study for CUSUM, and a Monte Carlo simulation for window-limited CUSUM and FMA, to assess their operating characteristics and the accuracy of the theoretical bounds/approximations.
Lastly, in Section~\ref{sec:conclusions} we draw conclusions and discuss other avenues for advancing the field of transient change detection.

\section{Change detection formulations} \label{sec:formulations}

Before we proceed to address the optimization problem associated with intermittent change detection, we provide a brief overview of different ways to look at a large class of change detection criteria.
We investigate the differences and common features of various approaches and justify why we focus on a particular one.

\subsection{The general model}

We start with a sequence of observations $Y_1, Y_2, \dots$.
At some unknown moment in time $\nu \in \{0,1,2,\dots\}$ the distribution of the observations may undergo a change that lasts for an unknown time $N$ (random or deterministic).
More specifically, $\nu$ denotes the serial number of the last pre-change observation, i.e., $Y_{\nu + 1}$ is the first under-change observation and $Y_{\nu + N}$ is the last under-change observation.
Let $\bY^k = (Y_1, \cdots, Y_k)$ denote the first $k$ observations and $\cF_k = \sigma(\bY^k)$ denote the corresponding filtration.

The joint density $p_\nu(\cdot \mid N = n)$ of the first $k$ observations is of the form
\begin{align} \label{eq:model:general}
\begin{split}
    p_\nu(\bY^k \mid N = n) &= \prod_{t = 1}^{k} g(Y_t \mid \bY^{t - 1})
        \quad \text{for $k \leq \nu$}, \\
    p_\nu(\bY^k \mid N = n) &= \prod_{t = 1}^{\nu} g(Y_t \mid \bY^{t - 1})
        \times \prod_{t = \nu + 1}^{k} f(Y_t \mid \bY^{t - 1})
        \quad \text{for $\nu < k \leq \nu + n$}, \\
    p_\nu(\bY^k \mid N = n) &= \prod_{t = 1}^{\nu} g(Y_t \mid \bY^{t - 1})
        \times \prod_{t = \nu + 1}^{\nu + n} f(Y_t \mid \bY^{t - 1})
        \times \!\!\! \prod_{t = \nu + n + 1}^{k} g(Y_t \mid \bY^{t - 1}) \\
        & \hspace{20em} \text{for $k > \nu + n$}, 
\end{split}
\end{align}
where $g(Y_t \mid \bY^{t - 1})$ is the pre-change and $f(Y_t \mid \bY^{t - 1})$ is the under-change conditional densities, respectively.

In what follows, we write $\Pr_{\nu}(\cdot \mid N = n)$ for the probability measure for the whole sequence $\{Y_t, t \geq 1\}$ for each last pre-change index $\nu \geq 0$ and change duration $N=n$, under which the observed sequence has density $p_\nu(\cdot \mid N = n)$ defined in \eqref{eq:model:general}.
Let $\EV_{\nu}(\cdot \mid N = n)$ denote the corresponding expectation.
We write $\Pr_\infty$ and $\EV_\infty$ for the probability measure and the corresponding expectation under the assumption that the change never happens ($\nu = \infty$).

A change detection procedure $T$ is a stopping time with respect to the filtration $\{\cF_k\}_{k \geq 1}$, i.e., $\{T = k\} \in \cF_k$. The ultimate goal is to find a
stopping time that is ``best'' at detecting the intermittent change subject to certain constraints.
We explore various approaches for measuring false alarms and correct detection in the next subsections to explain what exactly we understand by ``best'' and to properly formulate the optimization problem.

\subsection{Specific i.i.d.\ formulation} \label{sec:formulations:specifics}

Throughout the rest of the paper, we consider the case that observations are independent and each observation shares the same density $g$ under the nominal regime and a different density $f$ when the change is in effect (see the diagram below).
\[
    \underbrace{Y_1, \cdots , Y_\nu }_{\text{i.i.d., $g$}}, ~
    \underbrace{Y_{\nu + 1}, \cdots , Y_{\nu + N}}_{\text{i.i.d., $f$}}, ~
    \underbrace{Y_{\nu + N + 1}, \cdots }_{\text{i.i.d., $g$}}.
\]
In this special case, the general non-i.i.d.\ model given in \eqref{eq:model:general} takes the form
\begin{align} \label{eq:model:iid}
\begin{split}
    p_\nu(\bY^k \mid N = n) &= \prod_{t = 1}^{k} g(Y_t)
        \quad \text{for $k \leq \nu$}, \\
    p_\nu(\bY^k \mid N = n) &= \prod_{t = 1}^{\nu} g(Y_t)
        \times \prod_{t = \nu + 1}^{k} f(Y_t)
        \quad \text{for $\nu < k \leq \nu + n$}, \\
    p_\nu(\bY^k \mid N = n) &= \prod_{t = 1}^{\nu} g(Y_t)
        \times \prod_{t = \nu + 1}^{\nu + n} f(Y_t)
        \times \!\!\! \prod_{t = \nu + n + 1}^{k} g(Y_t)
        \quad \text{for $k > \nu + n$},
\end{split}
\end{align}
so that under the probability measure $\Pr_{\nu}(\cdot \mid N = n)$ observations $Y_1,\dots, Y_\nu$ and $Y_{\nu + N + 1}, Y_{\nu + N + 2}, \dots$ are i.i.d.\ with density $g$ and independent of $Y_{\nu + 1}, \dots, Y_{\nu + N}$, which are i.i.d.\ with density $f$.

The stopping times we consider are based on the maximal likelihood ratio approach. To facilitate their presentation, introduce the instantaneous likelihood ratios (LR)
\begin{align} \label{eq:instant_LR}
    \Lambda_k = f(Y_k) / g(Y_k) \qquad \text{for $k = 1, 2, \cdots $},
\end{align}
and the instantaneous log-likelihood ratios $\lambda_k = \log(\Lambda_k)$.

Although problem formulation in Sections \ref{sec:formulations:classes} and \ref{sec:formulations:optimization}, as well as the definition of the stopping times in Section \ref{sec:candidates}, applies to the general model \eqref{eq:model:general}, the specific form of their detection statistics as well as results presented in Sections \ref{sec:procedure_design} and \ref{sec:simulations} assume the i.i.d.\ model \eqref{eq:model:iid}.

\subsection{Classes of procedures} \label{sec:formulations:classes}

For the comparison between different rules to be fair, one restricts oneself to a certain class of procedures. We review three such classes.
Before we proceed, we introduce three conventional performance measures for an arbitrary stopping time $T$ under the no-change assumption: (i) average run length to false alarm ($\ARL$); (ii) local unconditional probability of false alarm ($\LUPFA_m$); and (iii) local conditional probability of false alarm ($\LCPFA_m$).
Specifically,
\begin{align*}
    \ARL(T) &= \EV_\infty (T), \\
    \LUPFA_m (T) &= \sup_{\ell \geq 0} \Pr_\infty (\ell < T \leq \ell + m), \\
    \LCPFA_m (T) &= \sup_{\ell \geq 0} \Pr_\infty (T \leq \ell + m \mid T > \ell).
\end{align*}

The first, average run length to false alarm ($\ARL$), commonly appears in persistent change-point detection ($N = \infty$) when one searches for an optimal stopping time in the class of stopping times
\begin{align} \label{eq:class:arl}
    \class_\gamma = \cset{T}{\ARL(T) \geq \gamma},
\end{align}
for some $\gamma \geq 1$.
However, this class is appropriate only if the $\Pr_\infty$-distribution of the stopping time $T$ is close to geometric, and may become inappropriate for intermittent change detection (see, e.g., \cite{Mei-SQA08, Tartakovsky-SQA08a,TNB_book2014,Tartakovsky_book2020}).

The following alternatives address this issue by focusing on constraining local probability of false alarm (unconditional and conditional, respectively):
\begin{align}
    \class^\star (m, \alpha) &= \cset{T}{\LUPFA_m (T) \leq \alpha}, \label{eq:class:lpfa_unconditional} \\
    \class (m, \alpha) &= \cset{T}{\LCPFA_m (T) \leq \alpha}, \label{eq:class:lpfa_conditional}
\end{align}
for some $m \geq 1$ and $0 < \alpha < 1$.

Let us look into the relationship between the three classes.
It turns out that all three are essentially different, with $\class_\gamma$ being the least restrictive and $\class(m, \alpha)$ the most stringent.

\begin{proposition}
    $\class(m, \alpha)$ is more stringent than $\class^\star (m, \alpha)$.
    In particular, if $T \in \class(m, \alpha)$, then $T \in \class^\star (m, \alpha)$
\end{proposition}
\begin{proof}
    It is not hard to see that $\class(m, \alpha) \subseteq \class^\star (m, \alpha)$ for every $m \geq 1$ and $0 < \alpha < 1$ since $\Pr_\infty(T \leq \ell + m \mid T > \ell) \geq \Pr_\infty(\ell < T \leq \ell + m)$ for all $\ell > 0$.

    Suppose now that $\class^\star (m, \alpha)$ can be covered by $\class(m', \alpha')$ for some $m' \geq 1$ and $0 < \alpha' < 1$.
    Consider a stopping time $T_r$ such that
    \begin{alignat*}{2}
        \Pr_\infty(T_r = i) &= 1 / (K + 1) \quad && \text{for $1 \leq i \leq K$}, \\
        \Pr_\infty(T_r = K + j) &= r \, (1 - r)^{j - 1} / (K + 1) \quad && \text{for $j \geq 1$},
    \end{alignat*}
    where $K = \ceiling{m / \alpha}$. Clearly, $T_r \in \class^\star (m, \alpha)$ for any $r$, $0 < r < 1$.
    However,
    \begin{align*}
        \LCPFA_{m'} (T_r) &\geq \sup_{\ell \geq K} \Pr_\infty(T_r \leq \ell + m' \mid T_r > \ell)
            = 1 - (1 - r)^{m'},
    \end{align*}
    and $r$ can be chosen so that $\LCPFA_{m'} (T_r) > \alpha'$, which leads to a contradiction.
\end{proof}

Now, when we consider $\class_\gamma$, it is intuitively obvious that large values of the ARL to false alarm do not necessarily guarantee small values of the maximal probability of false alarm, which has been discussed in \cite{Tartakovsky-SQA08a,TNB_book2014,Tartakovsky_book2020,Tartakovsky+et+al:2021}.
The following two propositions further highlight this relationship.

\begin{proposition}
    $\class(m, \alpha)$ is more stringent than $\class_\gamma$.
    In particular, if $T \in \class(m, \alpha )$, then $T \in \class_\gamma$ for
    \begin{align} \label{eq:arl_bound:gamma_from_m_alpha}
        \gamma - 1 = m \left( \frac{1}{\alpha} - 1\right).
    \end{align}
\end{proposition}

\begin{proof}
    Fix $m \geq 1$ and $0 < \alpha < 1$, and consider an arbitrary stopping time $T \in \class (m, \alpha)$:
    \begin{align} \label{eq:lemma:assumption}
        \sup _{\ell \geq 0} \Pr _\infty (T \leq \ell + m \mid T > \ell ) \leq \alpha .
    \end{align}
    Our goal is to show that $\EV _\infty (T) \geq \gamma$ for some $\gamma = \gamma(m, \alpha)$.
    With very small loss of generality consider only stopping times that are almost surely positive, i.e., $\Pr_\infty(T=0)=0$.

    Using induction
    \begin{align*}
        \Pr _\infty (T > i + k \, m \mid T > i)
        &= \Pr _\infty \left(T > [i + (k - 1) \, m] + m \mid T > [i + (k - 1) \, m]\right) \\
        &\quad \times \Pr _\infty \left(T > i + (k - 1) \, m \mid T > i + (k - 1) \, m\right),
    \end{align*}
    we obtain
    \begin{align} \label{eq:lemma:conditional_sequencing}
        \begin{split}
            \Pr _\infty (T > i + k \, m \mid T > i)  = \prod_{j = 0}^{k - 1} \Pr _\infty \left(T > [i + j m] + m \mid T > [i + j m]\right) \geq (1-\alpha)^k,
        \end{split}
    \end{align}
    where the last inequality follows from the fact that by \eqref{eq:lemma:assumption}
    \[
        \Pr _\infty (T > \ell + m \mid T > \ell ) \geq 1-\alpha \quad \text{for all} ~ \ell \geq 1.
    \]

    Now,
    \begin{align*}
        \EV _\infty (T) &= \sum _{\ell = 0} ^{\infty } \Pr _\infty (T > \ell ) \\
            &= \sum _{i = 0} ^{m - 1} \sum _{k = 0} ^{\infty } \Pr _\infty (T > i + k \, m) \\
            &= \sum _{i = 0} ^{m - 1} \Pr _\infty (T > i) \sum _{k = 0} ^{\infty } \Pr _\infty (T > i + k \, m \mid T > i).
    \end{align*}
    It follows from \eqref{eq:lemma:conditional_sequencing} that for the inner summation
    \begin{align*}
        \sum _{k = 0} ^{\infty } \Pr _\infty (T > i + k \, m \mid T > i)
            &\geq \sum _{k = 0}^{\infty } (1 - \alpha )^k
            = 1 / \alpha ,
    \end{align*}
    and for the outer summation
    \begin{align*}
        \sum _{i = 0} ^{m - 1} \Pr _\infty (T > i)
            &\geq 1 + \sum _{i = 1} ^{m - 1} \Pr _\infty (T > m) \\
            &\geq 1 + \sum _{i = 1} ^{m - 1} (1 - \alpha) = \alpha + m (1 - \alpha ).
    \end{align*}
    Combining these two inequalities together, we get that
    \begin{align}\label{ARLvsLCPFA}
        \EV _\infty (T) \geq \gamma(m, \alpha) \quad \text{for} ~ \gamma(m,\alpha) = 1 + m (1 - \alpha) / \alpha,
    \end{align}
    which completes the first part of the proof.

    On the other hand, if $T\in \class_\gamma$, then in general there is no $\alpha(m, \gamma)$ such that $\LCPFA_m (T) \leq \alpha(m,\gamma)$.
    Indeed, suppose that $\class_\gamma$ can be covered by $\class(m, \alpha)$ for some $m \geq 1$ and $0 < \alpha < 1$.
    Consider $T_r = \ceiling{\gamma} + G_r$, where $G_r$ is geometrically distributed with parameter $r$, $\alpha < r < 1$. Clearly, $T_r \in C_\gamma$. However,
    \[
        \LCPFA_m (T_r) \geq \sup_{\ell \geq \gamma} \Pr_\infty (T_r \leq \ell + m \mid T > \ell)
            = 1 - (1 - r)^m \geq r > \alpha,
    \]
    which contradicts our assumption that $\class_\gamma$ can be covered by $\class(m, \alpha)$ and completes the proof.
\end{proof}

\begin{proposition}
    $\class^\star(m, \alpha)$ is more stringent than $\class_\gamma$.
    In particular, if $T \in \class^\star(m, \alpha)$, then $T \in \class_\gamma$ for
    \begin{align} \label{eq:arl_bound:gamma_from_m_alpha_star}
        \gamma - 1 = \frac{m}{2} \left(\floor{\frac{1}{\alpha}} - 1\right).
    \end{align}
\end{proposition}

\begin{proof}
    To show that any $\class^\star (m, \alpha)$ can be covered by some $\class_\gamma$ fix $m \geq 1$ and $0 < \alpha < 1$ and let $T$ be such that
    \[
        \Pr_\infty(T = 1 + j m) = \begin{cases}
            \alpha \quad &\text{for $0 \leq j \leq K - 1$}, \\
            1 - K \alpha &\text{for $j = K$},
        \end{cases}
    \]
    where $K = \floor{1 / \alpha}$. Clearly, $T \in \class^\star (m, \alpha)$ and
    \begin{align*}
        \EV_\infty(T) = 1 + K m - K \, \frac{m \alpha}{2} - K^2 \, \frac{m \alpha}{2}
            \geq 1 + \frac{m}{2} \left(\floor{\frac{1}{\alpha}} - 1\right),
    \end{align*}
    which proves that $T \in \class^\star(m, \alpha)$ implies $T \in \class_\gamma$ with $\gamma = \gamma(m, \alpha)$ given by \eqref{eq:arl_bound:gamma_from_m_alpha_star}.

    Now, suppose that $\class_\gamma$ can be covered by $\class(m, \alpha)$ for some $m \geq 1$ and $0 < \alpha < 1$.
    As before, consider $T_r = \ceiling{\gamma} + G_r$, where $G_r$ is geometrically distributed with parameter $r$, $\alpha < r < 1$. Clearly, $T_r \in C_\gamma$. However,
    \[
        \LUPFA_m (T_r) = \Pr_\infty (T_r \leq \ceiling{\gamma} + m)
            = 1 - (1 - r)^m \geq r > \alpha,
    \]
    which contradicts the assumption and completes the proof.
\end{proof}

It is also worth pointing out that class $\class(m, \alpha)$ is the most natural choice for a practitioner.
Indeed, notice that a large value of the ARL to false alarm $\gamma$ does not necessarily guarantee small values of the maximal probabilities of false alarm, $\sup_\ell \Pr_\infty (\ell < T \leq \ell + m)$ and $\sup_\ell \Pr_\infty (T \leq \ell + m \mid T > \ell)$, as has been discussed by \cite{LaiIEEE98}, \cite{Tartakovsky-SQA08a}, and \cite{TNB_book2014} in detail.
Therefore, class $\class_\gamma$ may not be appropriate in a general case.
Now, assume that there is no change and consider a geometrically distributed stopping time\footnote{For example, the stopping times of the CUSUM and Shyriaev-Roberts (SR) detection procedures that start from the random initial conditions distributed according to quasi-stationary distributions are geometrically distributed.}, in which case there is a one-to-one correspondence between classes $\class_\gamma$ and $\class(m, \alpha)$.
Due to its memorylessness property, it is clear that as long as an alarm has not been raised, one can ignore the past observations and focus on detecting the change as if one just started the observation process.
This is in contrast to the unconditional $\class^\star(m, \alpha)$ that only takes into account the time when the procedure is initiated, disregarding the evolution of the observer.
Moreover, in this case, the unconditional probability $\Pr_\infty (\ell < T \leq \ell + m)$ achieves maximal value at $\ell=0$ and decays exponentially fast with $\ell$, while the conditional probability $\Pr_\infty (T \leq \ell + m \mid T > \ell)$ is constant.
So, it hardly makes sense to maximize unconditional probability which has its maximum at $\ell = 0$ and becomes almost $0$ after a handful of observations.
All aforementioned is true approximately in many cases where the properly normalized no-change distribution of stopping times is asymptotically exponential, which in the i.i.d.\ case is true for several detection procedures such as CUSUM and SR (see \cite{PollakTartakovskyTPA09}), SR mixtures when post-change parameters are unknown (see \cite{Yakir-AS95}), and for the generalized likelihood ratio CUSUM even in a substantially non-stationary case (see \cite{LiangTartakovskyVeerIEEEIT2022}).

The above argument allows us to conclude that class $\class(m, \alpha)$ is the most appropriate for most applications even if the ARL to a false alarm can be used as a measure of false alarms.

\subsection{The classical quickest change-point detection problem} \label{sec:clasQCPD}

As discussed in \cite{TNB_book2014}, there are several optimization criteria in the quickest change-point detection problems, which differ by available prior information and the definition of the false alarm rate. One popular minimax criterion was introduced by \cite{lorden-ams71} in his seminal paper:
\[
    \inf_{T \in \class_\gamma} \sup_{\nu \geq 0} \esssup \EV_{\nu}[T - \nu \mid T > \nu, Y_1, \dots, Y_\nu].
\]
It requires minimizing the conditional expected delay to detection $\EV_{\nu}[T - \nu \mid T > \nu, \cF_\nu]$ in the worst-case scenario with respect to both the change point $\nu$ and the trajectory $(Y_1, \dots, Y_\nu)$ of the observed process in the class of detection procedures $\class_\gamma$ with $\ARL(T) \geq \gamma$.
Hereafter $\esssup$ stands for essential supremum.
Lorden proved that Page's CUSUM detection procedure is asymptotically first-order minimax optimal as $\gamma \to \infty$.
Later on, \cite{MoustakidesAS86} in his ingenious paper established the exact optimality of CUSUM for any ARL to false alarm $\gamma \geq 1$.

Another popular, less pessimistic minimax criterion is due to \cite{PollakAS85}
\[
    \inf_{T \in \class_\gamma} \sup_{\nu \geq 0} \EV_{\nu}[T-\nu \mid T > \nu],
\]
which requires minimizing the conditional expected delay to detection $\EV_{\nu}[T-\nu \mid T > \nu]$ in the worst-case scenario with respect to the change point $\nu$ subject to a lower bound $\gamma$ on the ARL to false alarm.
\cite{PollakAS85} showed that the modified Shiryaev-Roberts (SR) detection procedure that starts from the quasi-stationary distribution of the SR statistic is third-order asymptotically optimal as $\gamma\to\infty$, i.e., the best one can attain up to an additive term $o(1)$, where $o(1)\to0$ as $\gamma\to\infty$.
Later \cite{tartakovskypolpolunch-tpa11} proved that this is also true for the SR-$r$ procedure that starts from the fixed but specially designed point $r$. See also \cite{PolunTartakovskyAS10} on exact optimality of the SR-$r$ procedure in a special case.

As mentioned in the introduction, the quickest change detection criteria may not be appropriate for the detection of transient changes of finite length $N$ in scenarios where detecting the change outside of the interval $[\nu + 1, \nu + N]$ leads to too large a loss associated with missed detection.
In the rest of the paper, we focus on reliable change detection.

\subsection{Reliable change detection optimization problem} \label{sec:formulations:optimization}

In what follows, we restrict ourselves to class $\class(m, \alpha)$ defined in \eqref{eq:class:lpfa_conditional}.
The task is to find a stopping time $T^\star \in \class (m, \alpha)$ that optimizes the chosen performance measure.
Just like there are different approaches for measuring the probability of false alarm, one can define the local probability of detection ($\LPD$) in several different ways.
Since we focus on \eqref{eq:class:lpfa_conditional}, we do not consider the unconditional local probability of detection considered, e.g., in \cite{Bakhache+Nikiforov:2000}, and only mention two alternatives.

The motivating performance measure introduced by~\cite{Tartakovsky_book2020} (Ch.\ 5) considers a random signal duration $N$ with known prior distribution $\pi=\{\pi_k\}_{k \geq 1}$, and was used, e.g., in \cite{Tartakovsky+et+al:2021}. Specifically, the problem is to find a stopping time $T$ that maximizes
\begin{align} \label{eq:lpd:lorden}
    \inf_{\nu \geq 0} \sum_{k = 1}^{\infty} \pi_k \essinf \Pr_{\nu} (\nu < T \leq \nu + k \mid \cF_\nu, N = k),
\end{align}
where $\essinf$ stands for essential infimum.
However, we adopt a different, less pessimistic measure
\begin{align} \label{eq:lpd:conditional}
    \LPD_\pi(T) &= \inf_{\nu \geq 0} \sum_{k = 1}^{\infty} \pi_k \Pr_{\nu} (T \leq \nu + k \mid T > \nu, N = k),
\end{align}
where $\pi$ is a valid pmf, corresponding either to the prior distribution of the signal duration if $N$ is random, or, in the case that $N$ is deterministic, to weights incorporating information about what $N$ could be, e.g., uniform weights would be appropriate if only the lower and upper bounds on the change duration are available.
Note that weighting detection probabilities associated with different change durations is essential even in the case that $N$ is deterministic: if we simply took infimum over all $\nu \geq 0$ and $N \geq 1$, the case of $N = 1$ would dominate the problem since $N = 1$ corresponds to the shortest (hence most difficult) change one could try to capture. More formally, this is a direct consequence of the fact that $\smallset{T \leq n} \in \cF_n$.

One's objective then is to solve the optimization problem
\begin{align} \label{eq:optimal_time}
    \sup _{T \in \class (m, \alpha )} \LPD_\pi(T).
\end{align}

One important practical problem is to maximize the probability of detection in some pre-specified window $M$ (possibly random) that does not exceed $N$, say detecting track initiation in a small window $M \ll N$.
In such scenarios, rather than setting $\pi$ to denote the pmf of the change duration in \eqref{eq:lpd:conditional}, one would set it to be the pmf of the detection window, since $\Pr_{\nu} (T \leq \nu + k \mid T > \nu, M = k \leq N) = \Pr_{\nu} (T \leq \nu + k \mid T > \nu, N = k)$.
Another possible motivation for the choice of $\pi$ in the deterministic $N$ case is mentioned in Section~\ref{sec:conclusions} (although we do not address it in this work).
Lastly, we would like to note that the case that $N$ is known has been considered by, e.g.,
\cite{Nikiforov+et+al:2012, Nikiforov+et+al:2017, Tartakovsky_book2020, Nikiforov+et+al:2023}.

In the remainder of the work, $\Durations$ denotes the essential support of $N$ if it is random, or $\Durations = \cset{k}{\pi_k > 0}$ otherwise.

\section{Maximal likelihood ratio-based rules} \label{sec:candidates}

In this section, we review several popular detection procedures that may appear to be nearly optimal when the probability of false alarm is small.
These procedures are designed, analyzed and compared in Sections~\ref{sec:procedure_design}~and~\ref{sec:simulations}.
However, all of them stem from the maximization of likelihood ratios paradigm in the context of intermittent change detection.
Specifically, we consider the following scenarios.
\begin{itemize}
    \item No information on the change duration is available (Section~\ref{subsec:uniform-ML}).
    \item An upper bound on the change duration is known (Section~\ref{subsec:window-limited-ML}).
    \item The change duration is known (Section~\ref{subsec:FMA}).
\end{itemize}
The three, respectively, give rise to the CUSUM rule; the so-called window-limited CUSUM (WL CUSUM) rule; and the Finite Moving Average (FMA) procedure.

\subsection{No information on change duration} \label{subsec:uniform-ML}
Consider the pessimistic case where no information about the change duration is available.
A natural choice of stopping time is based on the maximal likelihood ratio process for the intermittent change, with maximization over both the starting point and the endpoint of the change, with no constraints on the signal duration (except that it hadn't ended before the observation started):
\begin{align} \label{eq:iml_infty:statistic}
    \wtV_n = \max_{-\infty \leq k \leq n} \, \max_{\max\smallset{1, k} \leq \ell \leq n} \Bigg[ \sum_{i = k}^{\ell} \lambda_i \Bigg],
\end{align}
where $\lambda_i=\log \Lambda_i$ denotes the instantaneous log-likelihood ratio for the $i$th observation (see \eqref{eq:instant_LR}).
Maximization region is illustrated on the diagram below, where the highlighted region corresponds to the additional maximization points added when transitioning from $\wtV_{n - 1}$ to $\wtV_n$ (and also coincides with the maximization region for the traditional ``permanent-change'' CUSUM at step $n$).

\begin{center}
    \begin{tikzpicture}
        \pgfmathsetmacro{\xscale}{0.5}
        \pgfmathsetmacro{\yscale}{0.5}
        \draw[->] (0, 0) -- (\xscale * 6.5, 0) node[right] {starts at $\ensuremath{k}$};
        \draw[->] (0, 0) -- (0, \yscale * 6.5) node[above] {stops at $\ensuremath{\ell }$};
        \draw[dotted] (0, \yscale*5) -- (\xscale * 6.5, \yscale*5) node[right] {now};
        \foreach \x in {1,...,4}
            \foreach \y in {\x,...,4}
                \filldraw (\xscale*\x, \yscale*\y) circle (2pt);
        \foreach \x in {1,...,5} \filldraw[\ThemeRed] (\xscale*\x, \yscale*5) circle (2pt);
        \draw (\xscale*1, 0.1) -- (\xscale*1, -0.1); \draw (\xscale*1, -0.35) node {$\ensuremath{1}$};
        \draw (\xscale*2, 0.1) -- (\xscale*2, -0.1); \draw (\xscale*2, -0.35) node {$\ensuremath{2}$};
        \draw (\xscale*3.6, -0.35) node {$\ensuremath{\cdots }$};
        \draw (\xscale*5, 0.1) -- (\xscale*5, -0.1); \draw (\xscale*5, -0.35) node {$\ensuremath{n}$};
        \draw (0.1, \yscale*1) -- (-0.1, \yscale*1); \draw (-0.1, \yscale*1) node[left] {$\ensuremath{1}$};
        \draw (-0.1, \yscale*2.6) node[left] {$\ensuremath{\vdots }$};
        \draw (0.1, \yscale*4) -- (-0.1, \yscale*4); \draw (-0.1, \yscale*4) node[left] {$\ensuremath{n - 1}$};
        \draw (0.1, \yscale*5) -- (-0.1, \yscale*5); \draw (-0.1, \yscale*5) node[left] {$\ensuremath{n}$};
    \end{tikzpicture}
\end{center}
The corresponding stopping time is
\begin{align*} \nolabel{eq:iml_infty:time}
    \wtT = \wtT (b) &= \inf \smallcset{n \geq 1}{\wtV_n \geq b},
\end{align*}
where $b = b(m, \alpha )$ is chosen so that $\LCPFA _m (\wtT) = \alpha$.

Note that due to the unknown signal duration, the recursive form of \eqref{eq:iml_infty:statistic} is not as simple as in the conventional (when the change is persistent) CUSUM. Introducing the conventional CUSUM statistic
\begin{align}\label{eq:Vcusum}
    V_n &= \max \set{0, V_{n - 1}} +\lambda _n \qquad \text{for $n \geq 1$, with $V_0 = 0$},
\end{align}
it is not difficult to show that
\begin{alignat}{2}
    \wtV_n = \max \set{\wtV_{n - 1}, V_n} \qquad \text{for $n \geq 1$, with $\wtV_0 = 0$}.
\end{alignat}
In other words, the process $\smallset{\wtV_n}$ is nothing but the running maximum value of the CUSUM process $\smallset{V_n}$.
An immediate consequence is that the stopping time $\wtT$ coincides with the classical Page's CUSUM (\cite{page-bka54}) stopping time
\begin{align} \label{eq:cusum:time}
      \wtT (b)=  T_\CUSUM(b) = \inf \bigcset{n \geq 1}{\max _{1 \leq k \leq n} \sum _{i = k} ^{n} \lambda _i \geq b}
            = \inf \cset{n \geq 1}{V_n \geq b}.
\end{align}

\begin{proof}[Proof of recursion]
    Consider $\wtV_n$ for $n \geq 2$:
    \begin{align*}
        \wtV_n = \max _{1 \leq k \leq n} \max _{k \leq \ell \leq n} \sum _{i = k} ^{\ell } \lambda _i
            = \max \Bigset{\wtV_{n - 1}, \max _{1 \leq k \leq n} \sum _{i = k} ^{n} \lambda _i}.
    \end{align*}
    The first term covers all changes that ended at or before time $(n - 1)$. The second term takes into account all changes that ended at time $n$ and is
    nothing but the ``persistent-change'' CUSUM with a well-known recursion given by \eqref{eq:Vcusum}.
\end{proof}

\subsection{Upper bound on change duration} \label{subsec:window-limited-ML}
Suppose now that one knows that the change duration cannot exceed a given window size $M$, $M \geq 1$.
A reasonable course of actions in such circumstance is to maximize the likelihood ratio over changes that end after the observation starts with the signal duration constraint in mind:
\begin{align} \label{eq:iml_window:statistic}
    \wtV_{n : M} = \max_{-M + 2 \leq k \leq n} \, \max_{\max\smallset{1, k} \leq \ell \leq \min\smallset{n, k + M - 1}}
    \Bigg[ \sum_{i = k} ^{\ell } \lambda_i \Bigg].
\end{align}
Compared to \eqref{eq:iml_infty:statistic}, maximization happens over a strip-like region (rather than triangular), thus $\wtV_{n : M} \leq \wtV_n$ for all $n$.
\begin{center}
    \begin{tikzpicture}
        \pgfmathsetmacro{\xscale}{0.75}
        \pgfmathsetmacro{\yscale}{0.5}
        \pgfmathtruncatemacro{\breadth}{3}
        \pgfmathtruncatemacro{\special}{\breadth - 1}
        \pgfmathsetmacro{\markleft}{\xscale*(5 - \breadth + 0.5)}
        \pgfmathsetmacro{\markright}{\xscale*(5 + 0.5)}
        \draw[->] (-\xscale * 3.5, 0) -- (\xscale * 6.5, 0) node[right] {starts at $\ensuremath{k}$};
        \draw[->] (0, -0.3) -- (0, \yscale * 6.5) node[above] {stops at $\ensuremath{\ell }$};
        \draw[dotted] (0, \yscale*5) -- (\xscale * 6.5, \yscale*5) node[right] {now};
        \foreach \y in {1,...,\special} {
            \foreach \x in {1,...,\y}
                \filldraw (\xscale*\x, \yscale*\y) circle (2pt);
            \pgfmathtruncatemacro{\tempx}{-\breadth + \y + 1}
            \foreach \x in {\tempx,...,0}
                \draw[fill=white] (\xscale*\x, \yscale*\y) circle (2pt);
        }
        \foreach \y in {\breadth,...,4}
            \foreach \x in {1,...,\breadth}
                \filldraw ({\xscale*(\y - \breadth + \x)}, \yscale*\y) circle (2pt);
        \foreach \x in {1,...,\breadth} \filldraw[\ThemeRed] ({\xscale*(5 - \breadth + \x)}, \yscale*5) circle (2pt);
        \draw (-\xscale*1, 0.1) -- (-\xscale*1, -0.1); \draw (-\xscale*1 - 0.25, -0.35) node {$\ensuremath{-M + 2}$};
        \draw (\xscale*1, 0.1) -- (\xscale*1, -0.1); \draw (\xscale*1, -0.35) node {$\ensuremath{1}$};
        \draw (\xscale*2, 0.1) -- (\xscale*2, -0.1); \draw (\xscale*2, -0.35) node {$\ensuremath{2}$};
        \draw (\xscale*3.6, -0.35) node {$\ensuremath{\cdots }$};
        \draw (\xscale*5, 0.1) -- (\xscale*5, -0.1); \draw (\xscale*5, -0.35) node {$\ensuremath{n}$};
        \draw (0.1, \yscale*4) -- (-0.1, \yscale*4); \draw (-0.1, \yscale*4) node[left] {$\ensuremath{n - 1}$};
        \draw (0.1, \yscale*5) -- (-0.1, \yscale*5); \draw (-0.1, \yscale*5) node[left] {$\ensuremath{n}$};
        \draw[<->] (\markleft + 0.1, \yscale*5 + 0.5) -- (\markright - 0.1, \yscale*5 + 0.5);
        \draw (0.5*\markleft + 0.5*\markright, \yscale*5 + 0.5) node[above] {$\ensuremath{M}$};
        \draw[dashed] (\markleft, \yscale*5) -- (\markleft, \yscale*5 + 0.6);
        \draw[dashed] (\markright, \yscale*5) -- (\markright, \yscale*5 + 0.6);
    \end{tikzpicture}
\end{center}
Consequently, it leads to a different recursion:
\begin{align*}
    \wtV_{n : M} = \max\smallset{\wtV_{n - 1 : M}, V_{n : M}} \quad \text{for $n \geq 1$, with $\wtV_{0 : M} = 0$} ,
\end{align*}
where
\[
V_{n : M} = \max_{\max\smallset{1, n - M + 1} \leq k \leq n}\sum _{i = k} ^{n} \lambda _i.
\]

One can think of $\smallset{V_{n : M}}$ as a window-restricted version of the CUSUM statistic, $\smallset{V_n}$, although it no longer assumes a \emph{cumulative sum} representation. However, it may be implemented by re-running the CUSUM recursion on the latest $M$ observations with each new observation.
As with \eqref{eq:iml_infty:statistic}, the statistic \eqref{eq:iml_window:statistic} is non-decreasing.
The corresponding stopping time is
\begin{equation}\label{WLCUSUMst}
T_{\WLCUSUM : M} = \inf \smallcset{n \geq 1}{\wtV_{n : M} \geq b} = \inf \smallcset{n \geq 1}{V_{n : M} \geq b},
\end{equation}
where $b = b(m, \alpha )$ is chosen so that $\LCPFA _m (T_{\WLCUSUM : M}) = \alpha $.
It can be also rewritten as
\begin{align} \label{eq:iml_window:time}
    T_{\WLCUSUM : M} = \inf \Bigcset{n \geq 1}{\max _{\max\smallset{1, n - M + 1} \leq k \leq n} \sum _{i = k} ^{n} \lambda _i \geq b}.
\end{align}
A version of this rule, dubbed window-limited CUSUM (WL CUSUM), was proposed by \cite{willsky-ac76} and later on extensively studied in various settings including change-point detection, such as \cite{LaiIEEE98, Nikiforov+et+al:2012, Nikiforov+et+al:2017, Tartakovsky_book2020, Nikiforov+et+al:2023}.

\subsection{Change duration is known} \label{subsec:FMA}
In the case that we assume the change duration is known and equal to $M$, the detection statistic is
\begin{align} \label{eq:fma:statistic}
    \max_{-M + 2 \leq k \leq n} \left[ \sum_{i = \max\smallset{1, k}}^{k + M - 1} \lambda_i \right].
\end{align}

Note that if the change starts early on we try to capture it with truncated log-likelihood ratio, $\sum_{i = 1}^{n} \lambda_i$, $1 \leq n \leq M - 1$.
We would like to emphasize that we do \emph{not} consider changes shorter than $M$, just that for the first $M - 1$ observations we only see part of the whole change.
This contrasts with the previous cases (CUSUM and window-limited CUSUM), where at each moment in time there is at least one change scenario that is \emph{fully} within the observation bounds.
For this reason, we recommend that the first $M - 1$ thresholds be adjusted to compensate for the unobserved portion of the change (see figure below).
\begin{center}
    \begin{tikzpicture}
        \pgfmathsetmacro{\xscale}{1.0}
        \pgfmathsetmacro{\yscale}{0.5}
        \pgfmathtruncatemacro{\breadth}{3}
        \pgfmathtruncatemacro{\breadthLessOne}{\breadth - 1}
        \pgfmathsetmacro{\mark}{\xscale*(5 - \breadth + 1)}
        \draw[->] (-\xscale * 1.0 - 1.5, 0) -- (\xscale * 4.5, 0) node[right] {starts at $\ensuremath{k}$};
        \draw[->] (0, -0.3) -- (0, \yscale * 6.5) node[above] {stops at $\ensuremath{\ell }$};
        \draw[dotted] (0, \yscale*5) -- (\xscale * 4.5, \yscale*5) node[right] {now};
        \draw[dashed] (\mark, \yscale*5) -- (\mark, \yscale*0 + 0.3);
        \foreach \y in {\breadth,...,4}
            \foreach \x in {1,...,1}
                \filldraw ({\xscale*(\y - \breadth + \x)}, \yscale*\y) circle (2pt);
        \foreach \x in {1,...,1} \filldraw[\ThemeRed] ({\xscale*(5 - \breadth + \x)}, \yscale*5) circle (2pt);
        \foreach \y in {1,...,\breadthLessOne} {
            \pgfmathtruncatemacro{\tempx}{-\breadth + \y + 1}
            \draw[fill=white] (\xscale*\tempx, \yscale*\y) circle (2pt);
        }
        \draw (-\xscale*1, 0.1) -- (-\xscale*1, -0.1); \draw (-\xscale*1 - 0.25, -0.35) node {$\ensuremath{-M + 2}$};
        \draw (\xscale*1, 0.1) -- (\xscale*1, -0.1); \draw (\xscale*1, -0.35) node {$\ensuremath{1}$};
        \draw (\xscale*3, 0.1) -- (\xscale*3, -0.1); \draw (\xscale*3, -0.35) node {$\ensuremath{n - M + 1}$};
        \draw (0.1, \yscale*1) -- (-0.1, \yscale*1); \draw (-0.1, \yscale*1) node[left] {$\ensuremath{1}$};
        \draw (0.1, \yscale*\breadth) -- (-0.1, \yscale*\breadth); \draw (-0.1, \yscale*\breadth) node[left] {$\ensuremath{M}$};
        \draw (-0.1, \yscale*4.2) node[left] {$\ensuremath{\vdots}$};
        \draw (0.1, \yscale*5) -- (-0.1, \yscale*5); \draw (-0.1, \yscale*5) node[left] {$\ensuremath{n}$};
    \end{tikzpicture}
\end{center}
The resulting stopping time can be expressed by
\begin{align} \label{eq:fma:time}
    T_{\FMA : M} &= \inf \Bigcset{n \geq 1}{\hspace{-2.0em} \sum_{i = \max\smallset{1, n - M + 1}}^{n} \hspace{-2.0em} \lambda _i \geq b_n},
\end{align}
where $b_n = b$ for all $n \geq M$, and for $1 \leq n \leq M - 1$ are chosen such that
\[
    \Pr_\infty \Big(\sum_{i = 1}^{n} \lambda _i \geq b_n\Big) = \Pr_\infty \Big(\sum_{i = 1}^{M} \lambda _i \geq b\Big),
\]
or, equivalently,
\begin{align} \label{eq:fma_thresholds}
    b_n = H_n^{-1}\big(H_M(b)\big) \qquad \text{for $1 \leq n < M$},
\end{align}
where $H_n$ is the cdf of $\sum_{i = 1}^{n} \lambda _i$ under $\Pr_\infty$, which in most cases is either known or amenable to numerical calculation.
This rule is a modification of a rule which is commonly known as Finite Moving Average (FMA for short).
A standard alternative approach, which we advocate against but still consider in our simulations (Section~\ref{sec:simulations}), would be to skip the first $M - 1$ observations altogether (corresponding to the changes that could not be fully observed) and only stop at times $\geq M$.

The later standard FMA rule has been studied extensively, e.g., \cite{lai-as74} examined it from the quality control perspective, while \cite{Egea-Roca+et+al:2018} considered it in the context of intermittent signal detection.
In the special case that $f$ and $g$ are Gaussian the truncated version of FMA is equivalent to $\inf \smallcset{n \geq M}{\sum _{i = n - M + 1} ^{n} Y_i \geq c}$, and as such this simplified version has gained some attention as a more tractable problem by \cite{Noonan+Zhigljavsky:2020,Noonan+Zhigljavsky:2021}.

\section{Choosing thresholds} \label{sec:procedure_design}

One of the first steps in designing any stopping rule $T$ is to ensure that $T \in \class (m, \alpha)$, that is, its $\LCPFA(T) $ is upper-bounded by $\alpha$.

There are various approaches that one can take in order to choose the thresholds of these procedures to ensure that these rules are in class $\class(m, \alpha)$.
One such way takes advantage of the asymptotic distribution of the stopping times, which under the no-change hypothesis is often exponential (see, e.g., \cite{PollakTartakovskyTPA09}).
Suppose that under $\Pr _\infty $ a stopping time $T$ is geometrically distributed with parameter $\varrho$, at least approximately. That is,
\[
    \Pr _\infty (T = k) = \varrho \, (1 - \varrho )^{k - 1} \qquad \text{for $k = 1, 2, \cdots $.}
\]
It is not hard to see that in this case the expression $\Pr _\infty (T \leq \ell + m \mid T > \ell )$ does not depend on $\ell $, $\ell \geq 0$, and
\[
    \LCPFA _m (T) = 1 - (1 - \varrho )^m.
\]
Clearly, in this case, there is one-to-one correspondence between $\LCPFA _m (T)$ and $\ARL(T)$, and in order to ensure that $T \in \class(m, \alpha)$ it suffices to set
\begin{equation} \label{PFAvsARL}
    \ARL(T) = \frac{1}{1 - (1 - \alpha)^{1/m}}.
\end{equation}

Unfortunately, there is only a handful of cases when the no-change distribution of the stopping time is exactly geometric. Examples include the Shewhart procedure \cite[Sec 5.2]{Tartakovsky_book2020}, the randomized at $0$ CUSUM procedure when the CUSUM statistic $V_n$ starts not from $V_0=0$ but from the random value $V_0$ with the quasi-stationary distribution of the CUSUM statistic $\Pr(V_0 \leq v) = \lim_{n\to\infty} \Pr_\infty(V_n \leq v \mid T_{\rm CS}>n)$, and the Shiryaev-Roberts-Pollak procedure that starts from the quasi-stationary distribution of the Shiryaev-Roberts statistic.

Typically, $\Pr_\infty$-distributions of the properly normalized stopping times are asymptotically exponential.
For example, it follows from \cite{PollakTartakovskyTPA09} that $T_{\rm CS}/\ARL(T_{\rm CS})$ is asymptotically exponential as $\ARL(T_{\rm CS})$ is large.
Moreover, asymptotics kick in for moderate values of $\ARL$ which are reasonable for practical applications. Then if the estimate of $\ARL$ is available, the approximation of the $\LCPFA(T)$ can be easily found, e.g., using \eqref{PFAvsARL}.

However, such a general approach does not take advantage of the structure of the stopping time.
To this end, we review the three rules presented in the previous section and exploit their properties to propose more reliable and accurate design methods.

\subsection{CUSUM design} \label{subsec:design:CUSUM}

We begin with the numerical approach for designing the CUSUM detection procedure given by the stopping time \eqref{eq:cusum:time}, $T_\CUSUM = \inf \cset{n \geq 1}{V_n \geq b}$, where the CUSUM statistic satisfies the recursion
\begin{align*}
    V_n = \max\smallset{0, V_{n - 1}} + \lambda_{n}, \quad n \geq 1, \quad V_0 = 0.
\end{align*}
Note that $V_n$ is a homogeneous Markov process.
For such stopping times one does not have to rely on Monte-Carlo methods to choose design parameters.
Instead, one can adopt a numerical framework based on integral equations developed by \cite{MoustPolTarSS09}, which provides an accurate deterministic method of evaluating various operating characteristics with any desired precision.
We adapt it for CUSUM \eqref{eq:cusum:time} in order to get a handle on both $\LCPFA_m(T_\CUSUM)$ and $\LPD_\pi(T_\CUSUM)$.

To this end, consider a version of \eqref{eq:cusum:time} with its recursive statistic initialized at an arbitrary point $r < b$. Specifically,
\begin{align*}
    T_\CUSUM^\star(r) = \inf \cset{n \geq 1}{V_n^\star(r) \geq b},
\end{align*}
where
\[
    V_n^\star(r) = \max\smallset{0, V_{n - 1}^\star(r)} + \lambda_{n}, \quad V_0^\star(r) = r.
\]
Let $F_j$ denote the cdf of $\Lambda_1 = e^{\lambda_1}$ under $\Pr_j$ for $j = 0, \infty$, and introduce
\begin{alignat*}{2}
    \rho_{\ell, \infty}(r) &= \Pr_\infty(T_\CUSUM^\star(r) > \ell), \qquad & \rho_{0, \infty}(r) &\equiv 1, \\
    \rho_{\ell, \nu : n}(r) &= \Pr_{\nu} \big(T_\CUSUM^\star(r) > \ell \mid N = n\big), \qquad & \rho_{0, \nu : n}(r) &\equiv 1.
\end{alignat*}
Then, for $\ell \geq 1$, denoting $B = e^b$, one has
\begin{align} \label{eq:integral:tail:infty}
    \rho_{\ell, \infty}(r) &= \int_{0}^{B} \rho_{\ell - 1, \infty}(x) \, {\rm d} F_\infty \left(\frac{x}{\max\smallset{1, r}}\right),
\end{align}
and
\begin{align} \label{eq:integral:tail:change}
    \rho_{\ell, \nu : n}(r) &= \int_{0}^{B} \rho_{\ell - 1, \nu : n}(x) \, {\rm d} F_{J(\ell, \nu : n)} \left(\frac{x}{\max\smallset{1, r}}\right),
\end{align}
where $J(\ell, \nu : n) = 0$ if $\nu + 1 \leq \ell \leq \nu + n$, and $J(\ell, \nu : n) = \infty$ otherwise.
In order to find
\begin{align}
    \LCPFA_m(T_\CUSUM) &= 1 - \inf_{\ell \geq 0} \frac{\rho_{\ell + m, \infty}(1)}{\rho_{\ell, \infty}(1)} \label{eq:integral:lpfa}, \\
    \LPD_\pi (T_\CUSUM) &= 1 - \sup_{\ell \geq 0} \sum_{k \in \Durations} \pi_k \, \frac{\rho_{\ell + k, \ell : k}(1)}{\rho_{\ell, \ell : k}(1)} \label{eq:integral:lpd},
\end{align}
it suffices to solve the integral equations \eqref{eq:integral:tail:infty} and \eqref{eq:integral:tail:change}.
This can be achieved numerically by linearizing the system (cf.\ \cite{TartakovskyPolunchenko-FUSION08}).
Specifically, we partition the interval $[0, B]$ into $N$ sub-intervals with endpoints $0 = x_{0} < x_{1} < \cdots < x_{N} < x_{N + 1} = B$.
Denote the midpoints of these intervals $r_i = (x_{i} + x_{i + 1})/2$, and introduce $\hat{\rho}_{\ell, \infty}$ as piecewise-constant approximation to $\rho_{\ell, \infty}$ on $[0, B]$ (similarly for $\rho_{\ell, \nu : n}$):
\[
    \hat{\rho}_{\ell, \infty}(x) = \sum_{i = 0}^{N} \rho_{\ell, \infty}(r_i) \, \One \{x_i < x < x_{i + 1}\}.
\]
Substituting $\hat{\rho}$'s for $\rho$'s in \eqref{eq:integral:tail:infty} and \eqref{eq:integral:tail:change}, one obtains a system of linear equations for the values of $\rho$ at midpoints $r_i$, $1 \leq i \leq N$, that can be written compactly using matrix notation:
\begin{alignat}{2}
    \boldsymbol{\hat{\rho}}_{\ell, \infty} &= \boldsymbol{K}_\infty \cdot \boldsymbol{\hat{\rho}}_{\ell - 1, \infty}, \qquad & \boldsymbol{\hat{\rho}}_{0, \infty} &= \boldsymbol{1}, \label{eq:linear:tail:infty} \\
    \boldsymbol{\hat{\rho}}_{\ell, \nu : n} &= \boldsymbol{K}_{\! J(\ell, \nu : n)} \cdot \boldsymbol{\hat{\rho}}_{\ell - 1, \nu : n}, \qquad & \boldsymbol{\hat{\rho}}_{0, \nu : n} &= \boldsymbol{1}. \label{eq:linear:tail:change}
\end{alignat}
Here $\boldsymbol{K}_\infty$ and $\boldsymbol{K}_0$ are $N$-by-$N$ matrices whose $(i, j)$-th elements are
\[
    K_{i,j} = \Pr \left[\Lambda_1 < \frac{x_{j + 1}}{\max\smallset{1, r_i}}\right] - \Pr \left[\Lambda_1 < \frac{x_{j}}{\max\smallset{1, r_i}}\right],
    \qquad
\]
for $\Pr = \Pr_\infty$ and $\Pr = \Pr_0$, respectively; $\boldsymbol{1} = \left[1, 1, \dots, 1 \right]^{\top}$; and
\begin{align*}
    \boldsymbol{\hat{\rho}}_{\ell, \infty} &= \left[ \rho_{\ell, \infty}(r_1), \dots, \rho_{\ell, \infty}(r_N) \right]^{\top}, \\
    \boldsymbol{\hat{\rho}}_{\ell, \nu : n} &= \left[ \rho_{\ell, \nu : n}(r_1), \dots, \rho_{\ell, \nu : n}(r_N) \right]^{\top}.
\end{align*}
Thus, solving \eqref{eq:linear:tail:infty}--\eqref{eq:linear:tail:change} yields an approximate solution to \eqref{eq:integral:tail:infty}--\eqref{eq:integral:tail:change}.
Computational complexity of solving the system is minimal even for large values of $N$ and $\ell$ since it can be performed iteratively and only requires matrix-vector multiplication.

In order to get $\LCPFA_m(T_\CUSUM)$ and $\LPD_\pi (T_\CUSUM)$, we have to examine the functional dependency of ratio of $\rho$'s on $\ell$. This can be done numerically by capping $\ell$ at a high enough level (as further described in Section~\ref{sec:simulations}).
Numerical evidence suggests that infimum in \eqref{eq:integral:lpfa} is attained as $\ell \to \infty$ when CUSUM is in the quasi-stationary regime, while supremum in \eqref{eq:integral:lpd} is reached at $\ell = 0$.
It is worth noting that the quasi-stationary mode is attained relatively quickly, typically for $\ell$ on the order of dozens.
Consequently, for a given threshold and parameters $m$, $\pi_{k}$, we get an approximation for $\LCPFA_m(T_\CUSUM)$ and $\LPD_\pi (T_\CUSUM)$ through \eqref{eq:integral:lpfa} and \eqref{eq:integral:lpd}.
Since the process is not computationally expensive, finding the threshold for which the desired level of false alarms is attained can be achieved with relative ease.

In addition, the CUSUM procedure allows for an efficient asymptotic analysis, which can be used to obtain simple approximations for a relatively low false alarm rate.
Indeed, as established by \cite{PollakTartakovskyTPA09} $T_{\rm CS}(b)/\ARL(T_{\rm CS}(b))$ is asymptotically exponential as $b \to \infty$. Standard renewal-theoretic methods [cf. \cite{siegmund-book85,woodroofe-book82,Tartakovsky_book2020}] readily apply to obtain that
\[
    \ARL(T_{\rm CS}(b)) = C^{-1} e^b(1+o(1)) \quad \text{as}~~ b \to \infty,
\]
where the constant $C\in(0,1)$ depends on the model and can be computed explicitly by renewal-theoretic argument [cf. \cite{TartakovskyIEEECDC05,TNB_book2014}].
In particular, for the Gaussian model when the pre-change density $g$ is $(0, \sigma^2)$-normal and the under-change density $f$ is $(\mu, \sigma^2)$-normal this constant can be easily computed numerically from the formula
\begin{align}\label{Cnormal}
    C = \frac{2}{q} \exp\set{-2 \sum_{t=1}^\infty \frac{1}{t} \Phi\left( -\frac{1}{2} \sqrt{q t}\right)},
\end{align}
where $q=\mu^2/\sigma^2$ is the ``signal-to-noise ratio'' and $\Phi(x) = (2\pi)^{-1/2} \int_{-\infty}^x \exp\set{-t/2} \rm{d} t$ is the standard normal cdf.
Also, in this case, simple Siegmund's corrected Brownian motion approximation [\cite{siegmund-book85}] yields
\[
    C \approx \exp\set{-\rho \sqrt{q}}, \quad \rho = 0.582597,
\]
which is sufficiently accurate as long as $q$ is not large (typically for $q \leq 2$). See, e.g., Table 3.1, Sec 3.1.5 in \cite{TNB_book2014}.
Therefore, for sufficiently large (typically for moderate) threshold values we have an approximation
\begin{align}\label{LPFAapprox}
    \LPFA_m(T_{\rm CS}(b)) \approx 1- \exp\set{-m C/e^b}.
\end{align}

Next, by nonlinear renewal theory, the limiting $\Pr_0$-distribution of the stopping time $\tau_b=(T_{\rm CS}(b)-b/\mu)\sqrt{b \sigma^2/\mu^3}$, where $\mu = \EV_0[Y_1]$, $\sigma^2 = \Var_0(Y_1)$, is normal:
\[
    \Pr_0\set{\tau_b \leq x} = \Phi(x) \quad \text{as}~~ b \to \infty ~~\text{for all}~ x\in(-\infty, \infty)
\]
(see, e.g., Theorem 2.6.2 in \cite{TNB_book2014}).
Hence, the probability of detection when $\nu=0$ can be approximated as follows:
\begin{align}\label{PD0approx}
    \Pr_0(T_{\rm CS}(b) \leq k \mid N = k) \approx \Phi\Bigg(\frac{k}{\sqrt{b \sigma^2/\mu^3}} + \frac{b}{\mu}\Bigg).
\end{align}
We conjecture that for a variety of i.i.d.\ models
\[
    \inf_{\nu \geq 0} \Pr_\nu(T_{\rm CS}(b) \leq \nu + k \mid T_{\rm CS}(b) > \nu, N = k) = \Pr_0(T_{\rm CS}(b) \leq k \mid N=k),
\]
which is confirmed by numerical study for the Gaussian example in Section~\ref{sec:simulations}. Thus, using \eqref{PD0approx}, we obtain the
following approximation for the minimal probability of detection
\begin{align}\label{PDapprox}
    \LPD_\pi(T_{\rm CS}(b)) \approx \sum_{k\in \Durations} \pi_k \Phi\Bigg(\frac{k}{\sqrt{b \sigma^2/\mu^3}} + \frac{b}{\mu}\Bigg).
\end{align}

\subsection{Window-limited CUSUM design} \label{subsec:design:WL-CUSUM}

Unfortunately, the integral equations method described above only works when the detection statistic has Markovian nature.
Window-limited CUSUM, given in \eqref{eq:iml_window:time},
\begin{align*}
    T_{\WLCUSUM : M} = \inf \Bigcset{n \geq 1}{\max _{\max\smallset{1, n - M + 1} \leq k \leq n} \sum _{i = k} ^{n} \lambda _i \geq b},
\end{align*}
does not have that property, and one has to adopt different approaches.
In particular, we are going to obtain a set of bounds on error probabilities that would allow one to: (i) control probability of false alarm; and (ii) establish a lower bound on probability of detection.

A lower bound on $\LPD$ and upper bounds on the unconditional probability of false alarm, $\LUPFA$, for a version of WL CUSUM \eqref{eq:iml_window:time} has been established in \cite{Nikiforov+et+al:2017, Nikiforov+et+al:2023} in the case of known signal duration.
These results can be extended and modified to fit our setting.
While the idea behind the proof of the following results is not new, there are several key differences.
First, the stopping times considered by \cite{Nikiforov+et+al:2017, Nikiforov+et+al:2023} all had the delay of $M$ observations. Specifically, their versions of WL CUSUM and FMA could not stop during the ``warm-up'' period before the first $M$ observations have been collected.
Consequently, the optimization criterion used in that work maximized the local probability of detection over changes that start at least at time $M$, in our notation that would correspond to $\nu \geq M$. Our results hold for all three procedures (WL CUSUM, FMA, and modified FMA) with no restriction on when the change starts. That distinction is especially important for FMA when its window size is chosen to be greater than the smallest change duration, $\inf \Durations$.
Second, we dropped the need to check whether the partial sums of likelihood ratios are associated.
\begin{lemma} \label{lemma_3}
    For the WL CUSUM rule given by \eqref{eq:iml_window:time} the following bounds hold:
    \begin{align}
        \LCPFA_m(T_\WLCUSUM) &\leq 1 - \Big[\prod_{k = 1}^{M} \Pr_\infty \Big(\sum_{i = 1}^{k}\lambda_i < b\Big) \Big]^m, \label{eq:iml_window:lpfa_bound} \\
        \LPD_\pi (T_\WLCUSUM) &\geq \sum_{k \in \Durations} \pi_k \Pr_{0} \Big(\sum_{i = 1}^{k \wedge M}\lambda_i \geq b \mid N = k\Big), \label{eq:iml_window:lpd_bound}
    \end{align}
    where $k \wedge M = \min\smallset{k, M}$.
\end{lemma}
\begin{proof}
    Let $S_a^b = \sum_{i = a}^{b} \lambda_i$ for $1 \leq a \leq b$.
    Note that for the i.i.d.\ model \eqref{eq:model:iid} the partial sums $S_a^b$ are associated \cite[Thm.~5.1]{Robbins1954, esary} under $\Pr_\infty$.
    We first prove \eqref{eq:iml_window:lpfa_bound}.
    Consider $\LCPFA_m$:
    \begin{align*}
        \LCPFA_m(T_\WLCUSUM) &= 1 - \inf_{\ell \geq 0}\frac{\Pr_\infty(T_\WLCUSUM > \ell + m)}{\Pr_\infty(T_\WLCUSUM > \ell)} \\
            &\leq 1 - \inf_{\ell \geq 0} \Pr_\infty \Big( \bigcap_{n = \ell + 1}^{\ell + m} \Bigset{\max_{1 \leq j \leq n \wedge M} S_{n - j + 1}^{n} < b} \Big) \\
            &\leq 1 - \min_{0 \leq \ell \leq M - 1} \prod_{n = \ell + 1}^{\ell + m} \prod_{j = 1}^{n \wedge M} \Pr_\infty \Big(S_{n - j + 1}^{n} < b\Big) \\
            &= 1 - \Big[\prod_{j = 1}^{M} \Pr_\infty \Big(S_{M - j + 1}^{M} < b\Big) \Big]^m,
    \end{align*}
    where both inequalities are due to association between partial sums. The statement \eqref{eq:iml_window:lpfa_bound} follows due to independence of $\lambda$'s.
    Second, we prove \eqref{eq:iml_window:lpd_bound}.
    Clearly,
    \[
        \Pr_\nu(T_\WLCUSUM > \nu + k \mid T > \nu, N = k) = \frac{\Pr_\nu \Big(\bigcap_{n = 1}^{\nu + k} \Bigset{\displaystyle\max_{1 \leq j \leq n \wedge M} S_{n - j + 1}^{n} < b} \mid N = k \Big)}{\Pr_\infty \Big(\bigcap_{n = 1}^{\nu} \Bigset{\displaystyle\max_{1 \leq j \leq n \wedge M} S_{n - j + 1}^{n} < b} \Big)}.
    \]
    Let $(k - M)^+$ denote $\max\smallset{0, k - M}$. It is not hard to see that
    \[
        \bigcap_{n = 1}^{\nu + k} \Bigset{\max_{1 \leq j \leq n \wedge M} S_{n - j + 1}^{n} < b}
            \subseteq \bigcap_{n = 1}^{\nu} \Bigset{\max_{1 \leq j \leq n \wedge M} S_{n - j + 1}^{n} < b}
            \bigcap \Bigset{S_{\nu + 1 + (k - M)^+}^{\nu + k} < b},
    \]
    if all changes ending after $\nu$ are omitted except for the longest one that starts after $\nu$ and ends at exactly $\nu + k$. Due to independence of $\lambda$'s
    \begin{align*}
        &\Pr_\nu \Big(\bigcap_{n = 1}^{\nu + k}\Bigset{\max_{1 \leq j \leq n \wedge M} S_{n - j + 1}^{n} < b} \mid N = k \Big) \\
            &\qquad\leq
            \Pr_\infty \Big(\bigcap_{n = 1}^{\nu} \Bigset{\max_{1 \leq j \leq n \wedge M} S_{n - j + 1}^{n} < b}\Big)
            \Pr_\nu \Big(S_{\nu + 1 + (k - M)^+}^{\nu + k} < b \mid N = k\Big).
    \end{align*}
    Consequently,
    \[
        \Pr_\nu(T_\WLCUSUM > \nu + k \mid T > \nu, N = k) \leq \Pr_0 \Big(S_{1 + (k - M)^+}^{k} < b \mid N = k\Big),
    \]
    and \eqref{eq:iml_window:lpd_bound} follows.
\end{proof}

It is worth noting that for practical purposes the upper bound on the probability of false alarms \eqref{eq:iml_window:lpfa_bound} is most important, especially in security-critical applications where one needs to guarantee that the error probability does not exceed a prescribed level.
These bounds are amenable to numerical calculation and can be evaluated with arbitrary precision without relying on randomized methods.
We further explore the sharpness of both bounds in Section~\ref{sec:simulations}, although \eqref{eq:iml_window:lpd_bound} should be a lot less accurate since it relies on omitting a noticeable chunk of observation.

\subsection{FMA design} \label{subsec:design:FMA}

Recall the generalized FMA rule given in \eqref{eq:fma:time}:
\begin{align*}
    T_{\FMA : M} &= \inf \Bigcset{n \geq 1}{\hspace{-2.0em} \sum_{i = \max\smallset{1, n - M + 1}}^{n} \hspace{-2.0em} \lambda _i \geq b_n},
\end{align*}
where the first $(M - 1)$ thresholds are chosen according to \eqref{eq:fma_thresholds}.
This rule is similar to window-limited CUSUM studied in the previous section in that it has a sliding window structure to its detection statistic.
For this reason, bounds similar to \eqref{eq:iml_window:lpfa_bound}--\eqref{eq:iml_window:lpd_bound} can be established for FMA \eqref{eq:fma:time}, although the set of permissible change duration $\Durations$ directly affects the choice of $M$.

\begin{lemma} \label{lemma_4}
    For the FMA rule given by \eqref{eq:fma:time}--\eqref{eq:fma_thresholds} the following upper bound holds:
    \begin{align}
        \LCPFA_m(T_\FMA) &\leq 1 - \Big[\Pr_\infty \Big(\sum_{i = 1}^{M}\lambda_i < b\Big) \Big]^m.
        \label{eq:fma:lpfa_bound}
    \end{align}
    Furthermore, if $M \leq \inf \Durations$, then
    \begin{align}
        \LPD_\pi (T_\FMA) &\geq \Pr_{0} \Big(\sum_{i = 1}^{M} \lambda_i \geq b \mid N = M\Big). \label{eq:fma:lpd_bound}
    \end{align}
\end{lemma}

Although the proof follows the same lines as that for window-limited CUSUM, it is worth pointing out several key differences.
If there are values in $\Durations$ smaller than $M$, the lower bound on $\LPD_\pi$ is negatively affected and has a different form than in \eqref{eq:fma:lpd_bound}.
Indeed, the worst-case performance is then determined by the shorter changes where partial sums contain terms from both under-change and pre-change modes.
For this reason, we strongly recommend setting $M = \inf \Durations$.
Finally, we would like to mention that bounds similar to \eqref{eq:fma:lpfa_bound} and \eqref{eq:fma:lpd_bound} can be obtained for the version of FMA where all thresholds are the same, i.e., $b_n = b$ for all $n \geq 1$.

\section{Numerical study and simulations} \label{sec:simulations}

In this section, we carry out a comparative analysis of several detection rules. The four procedures of interest are: CUSUM \eqref{eq:cusum:time}; window-limited CUSUM \eqref{eq:iml_window:time}; classical FMA \eqref{eq:fma:time} with unadjusted thresholds; and modified FMA (further everywhere mFMA) with thresholds adjusted according to \eqref{eq:fma_thresholds}.

Recall the maximum likelihood-based nature of the rules we consider (see Section~\ref{sec:candidates}).
Since window-limited CUSUM assumes that an upper bound on the change duration is known, in the context of this work it makes sense to set its window size equal to the maximum of all possible change duration values, i.e., $M = \sup \Durations$.
FMA rules, on the other hand, arise from the assumption that the change duration is known.
Furthermore, Lemma~\ref{lemma_4} suggests that its probability of detection might be significantly affected when the actual change duration is smaller than the assumed putative value.
For that reason, we set the window size equal to the smallest possible change duration value, i.e., $M = \inf \Durations$.
In Subsection \ref{subsec:sim:performance}, we explore how reasonable the proposed choice of window size is.

Our task is two-fold.
\begin{itemize}
    \item In Subsection \ref{subsec:sim:performance}, we compare how the probability of detection \eqref{eq:class:lpfa_conditional} vs.\ the probability of false alarm \eqref{eq:lpd:conditional} varies between the procedures in question.
        In addition to the comparative analysis of four detection procedures, we check how accurate the theoretical bounds from Lemma \ref{lemma_3} and Lemma \ref{lemma_4} are for the operating characteristics of window-limited CUSUM and mFMA.

    \item In Subsection \ref{subsec:sim:arl}, we consider whether the distribution of the stopping times under $\Pr_\infty$ for each rule is close to exponential. To do this, we examine the QQ-plots for all the rules. If $\Pr_\infty$-distribution of all rules is approximately exponential, we can use the approximation \eqref{PFAvsARL} relating $\LCPFA_{m}$ to $\ARL$. We check how accurate this approximation is for each detection rule.
        We also consider two approximations for the ARL of the classical FMA: Lai's asymptotic approximation (\cite{lai-as74}) and Noonan and Zhigljavsky's continuous-time approximation (\cite{Noonan+Zhigljavsky:2020, Noonan+Zhigljavsky:2021}), which does not claim to be asymptotic. Thus, in addition to assessing the accuracy of said approximations, we also investigate the asymptotic behavior of the two.
\end{itemize}

\subsection{Performance analysis} \label{subsec:sim:performance}

In this section, we consider the case where the standard Gaussian i.i.d.\ sequence undergoes a shift in mean of $1$. Specifically,
\begin{align*}
    Y_t &\sim \cN(0, 1) \qquad \text{for $t \leq \nu$}, \\
    Y_t &\sim \cN(1, 1) \qquad \text{for $\nu + 1 \leq t \leq \nu + N$}, \\
    Y_t &\sim \cN(0, 1) \qquad \text{for $t < \nu + N$}.
\end{align*}
It is worth noting that this model has been verified for near-Earth space informatics when it is necessary to identify streaks of low observable space objects with telescopes in plain images that appear and disappear at unknown points in time and space \cite{TartakovskyetalIEEESP2021}.

We consider two cases for the potential change duration:
\begin{enumerate}
    \item varying from $5$ to $10$, i.e., $N \in \Durations = \smallset{5, 6, \cdots, 10}$; and
    \item varying from $7$ to $15$, i.e., $N \in \Durations = \smallset{7, 8, \cdots, 15}$.
\end{enumerate}
In the first case, we set $\LCPFA$ window to $m = 10$ and in the second case $m = 15$. In both cases,
we assume a uniform prior $\pi$ in $\LPD$ \eqref{eq:lpd:conditional}.

To obtain the operating characteristics of CUSUM we use the integral equations framework described above with $10^4$--by--$10^4$ matrices, which provides extreme precision.
Detection statistics of other procedures do not allow for recursive expression, so we fall back to Monte Carlo simulations.
For a valid comparison of the rules, the threshold for each procedure is chosen so that
\[
    \LCPFA_m(T_\CUSUM) \approx \LCPFA_m(T_\WLCUSUM) \approx \LCPFA_m(T_\FMA) \approx \LCPFA_m(T_{\MFMA}).
\]

When estimating $\LCPFA_m (T)$ we must find out where the supremum of $\Pr_\infty(T \leq \ell + m \mid T > \ell)$ as a function of $\ell$ is attained. Figure~\ref{fig:LPFA_vs_ell} shows $\Pr_\infty(T \leq \ell + m \mid T > \ell)$ as a function of $\ell$ for all four rules for the case $\Durations = \smallset{5, 6, \cdots, 10}$ and $m = 10$.
Note that for the CUSUM and the WL CUSUM algorithms the maximum is attained in the quasi-stationary mode, as expected.
For the FMA algorithm it peaks at $M - 1$, while for mFMA this is no longer the case, although its maximum is attained early on.

\begin{figure}
    \centering
    \includegraphics[width=\textwidth]{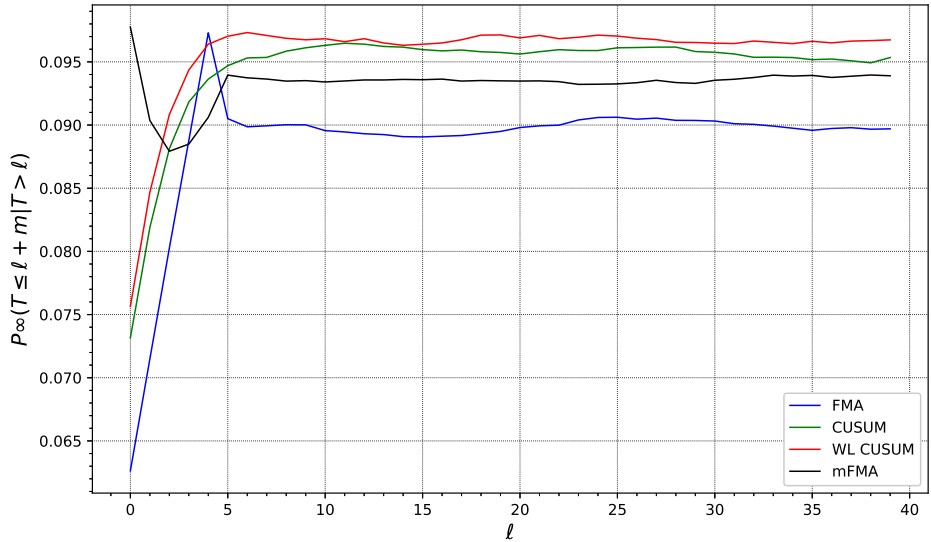}
    \caption{The PFA $\Pr_\infty(T \leq \ell + m \mid T > \ell)$ vs $\ell$ with $m = 10$.}
    \label{fig:LPFA_vs_ell}
\end{figure}

For each rule $T$ in question (except for CUSUM), for estimating $\LCPFA_m (T)$ we use Monte Carlo simulations to estimate its (unconditional) pmf.
To ensure accuracy, we run the simulations until at least $1000$ observations have been collected for each of the events $\set{T = j}$, $j \leq 30$.
Clearly,
\[
    \hat{p}_j = \frac{1}{K} \sum_{k = 1}^{K} \One \set{T_k > j}
\]
are unbiased estimators of $\Pr_\infty(T > j)$. We use $\hat{p}_{j + m} / \hat{p}_j$ to estimate $\Pr_\infty(T > j + m \mid T > j)$.
We use second (resp.\ first) order Taylor approximation of $\hat{p}_{j + m} / \hat{p}_j$ about $\Pr_\infty(T > j + m \mid T > j)$ to estimate its expected value (resp.\ variance).
Since $\Cov_\infty(\hat{p}_j, \hat{p}_{j + m}) = \Pr_\infty(T > j + m) \Pr_\infty(T \leq j) / K$ these approximations yield
\begin{align*}
    \EV_\infty \Big(\frac{\hat{p}_{j + m}}{\hat{p}_j}\Big) &\approx \Pr_\infty(T > j + m \mid T > j), \\
    \Var_\infty \Big(\frac{\hat{p}_{j + m}}{\hat{p}_j}\Big) &\approx \frac{\Pr_\infty(T > j + m \mid T > j) \, \Pr_\infty(T \leq j + m \mid T > j)}{K \, \Pr_\infty(T > j)},
\end{align*}
so we do not need to correct for bias and have a handle on the standard error.
The standard error varies depending on the values of $\LCPFA_m (T)$.
However, the relative values are almost the same and for different procedures do not exceed the following values:
\begin{itemize}
    \item for the first case: WL CUSUM -- $1 \%$, FMA -- $1 \%$, mFMA -- $0.9 \%$; and
    \item for the second case: WL CUSUM -- $0.8 \%$, FMA -- $0.8 \%$, mFMA -- $0.7 \%$.
\end{itemize}

To estimate $\LPD_\pi$ we must find out at what $\nu$ in the equation \eqref{eq:lpd:conditional} the infimum is attained.
Simulations (for WL CUSUM and FMA) and numerical analysis (for CUSUM) suggest that minimum in \eqref{eq:lpd:conditional} is reached when $\nu = 0$ for all procedures except for mFMA, for which it is attained before $\nu = 5$ in both scenarios (see Figure~\ref{fig:LPD_vs_nu} for the case $\Durations = \smallset{5, 6, \cdots, 10}$).

\begin{figure}[h]
    \centering
    \begin{subfigure}{0.49\columnwidth}
        \includegraphics[width=\textwidth]{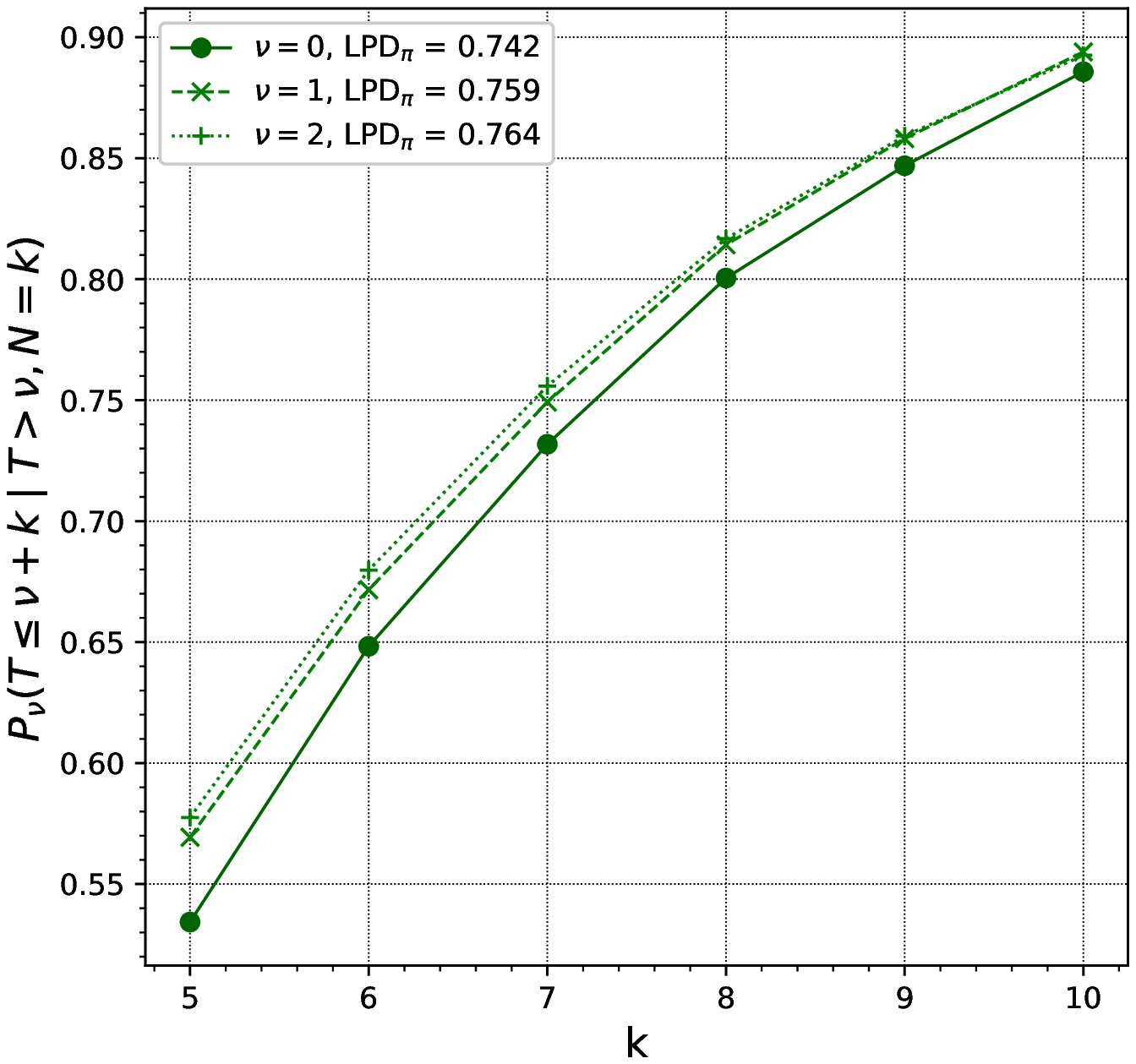}
        \caption{CUSUM}
    \end{subfigure}
    \begin{subfigure}{0.49\columnwidth}
        \includegraphics[width=\textwidth]{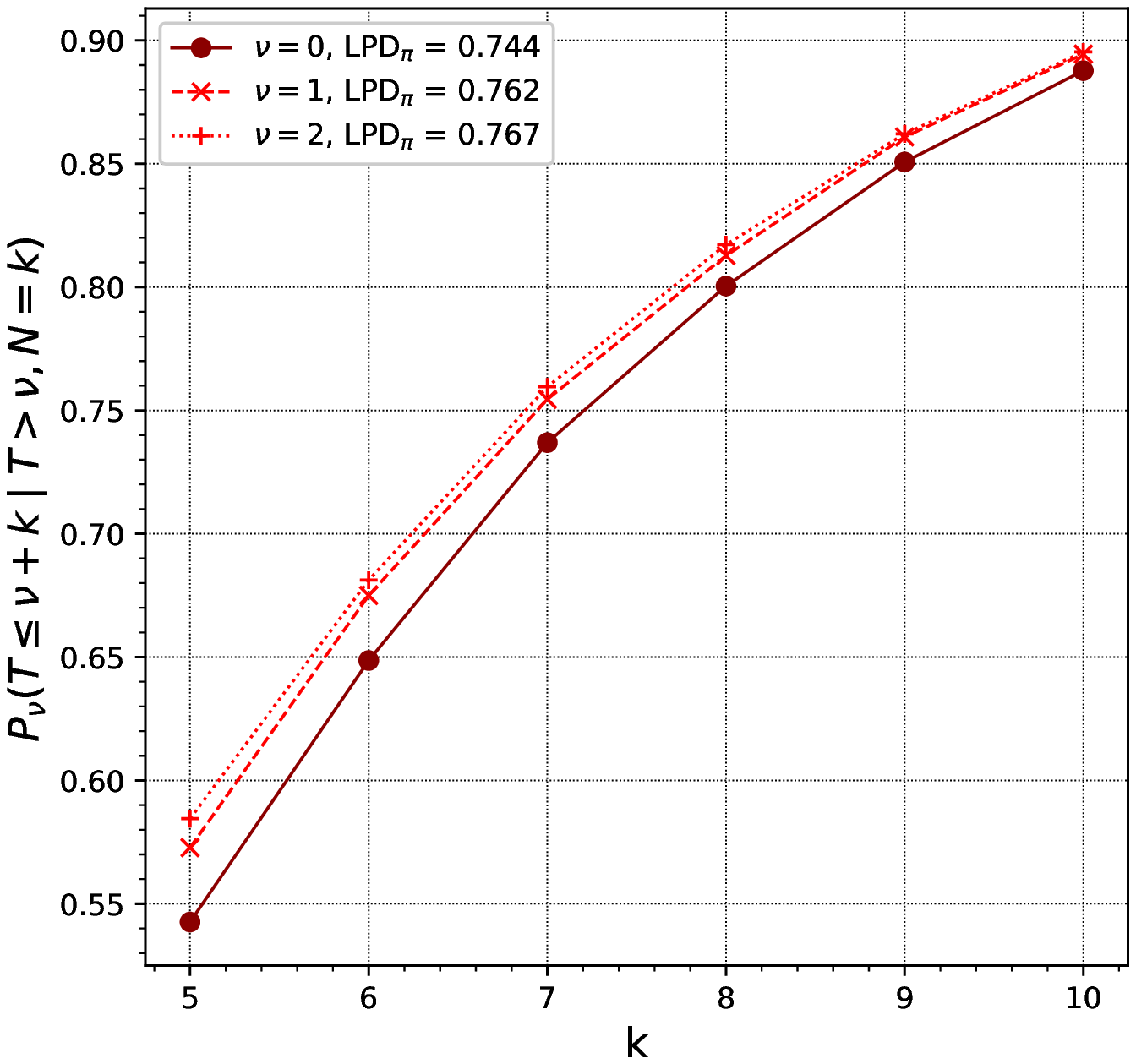}
        \caption{WL CUSUM}
    \end{subfigure}
    \begin{subfigure}{0.49\columnwidth}
        \includegraphics[width=\textwidth]{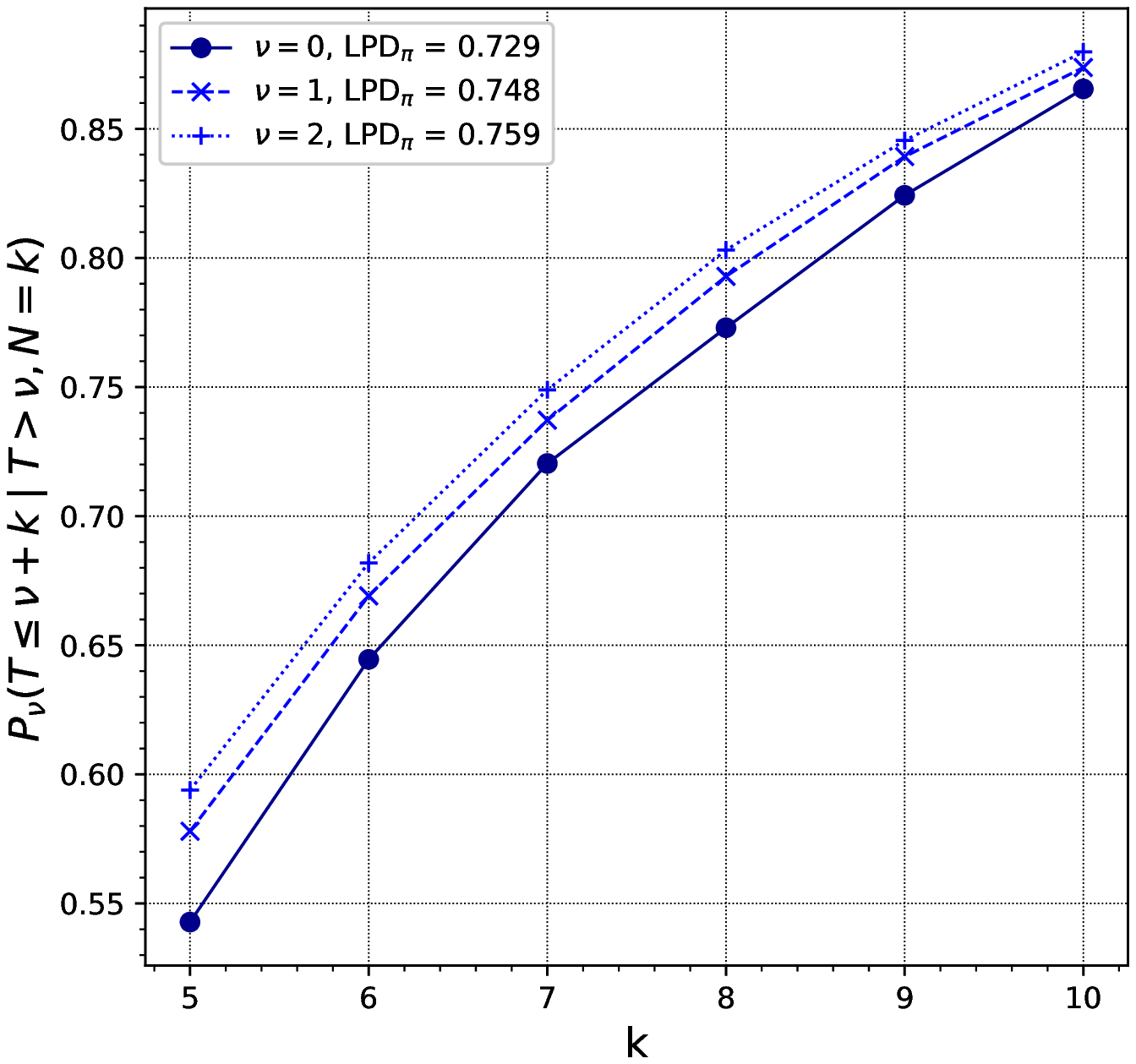}
        \caption{FMA}
    \end{subfigure}
    \begin{subfigure}{0.49\columnwidth}
        \includegraphics[width=\textwidth]{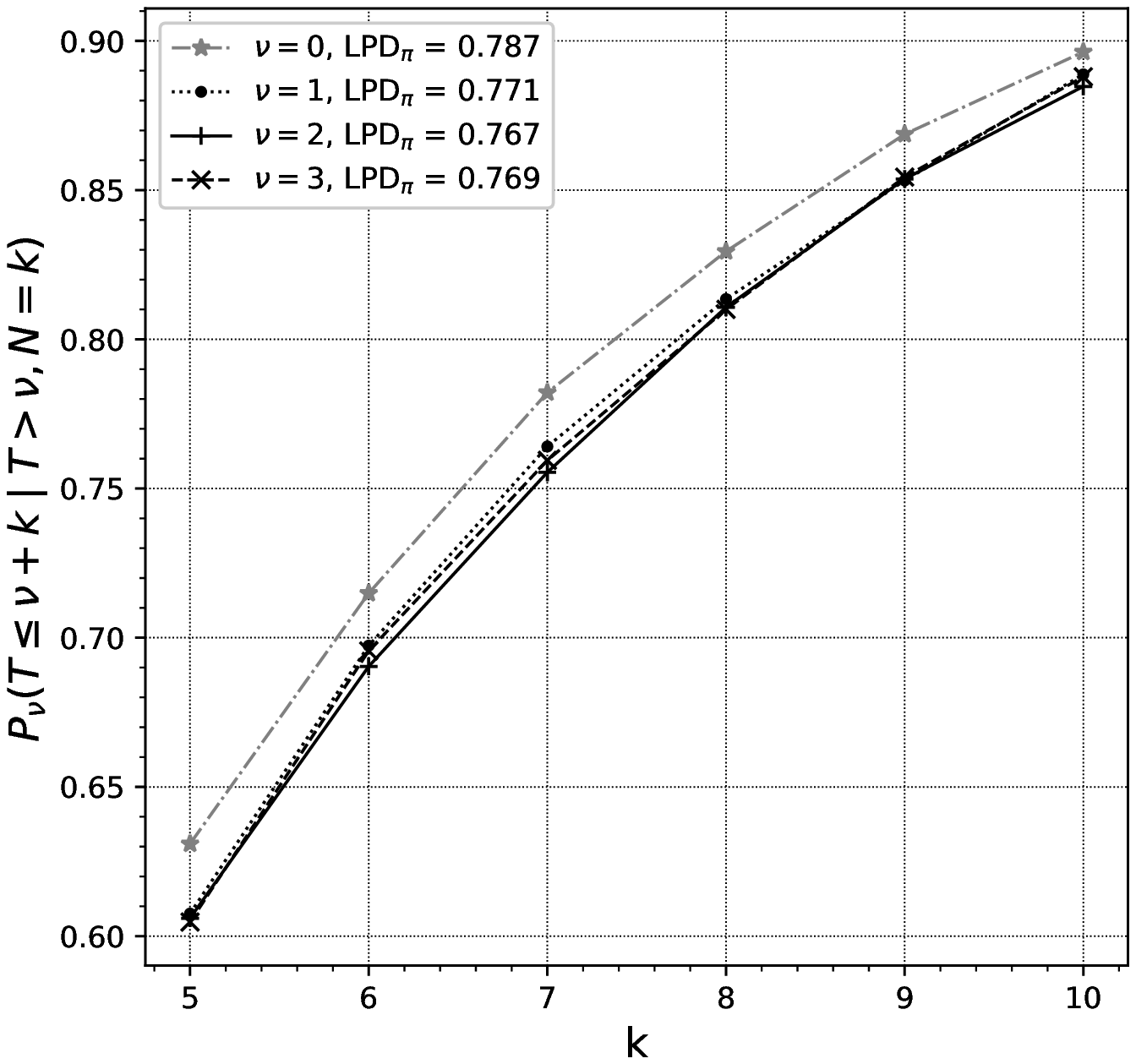}
        \caption{mFMA}
    \end{subfigure}
    \caption{Probability $\Pr_{\nu} (T \leq \nu + k \mid T > \nu, N = k)$ vs $k$ for the four detection rules with different $\nu$ when $\LCPFA_m \approx 0.1$.}
    \label{fig:LPD_vs_nu}
\end{figure}

Figure~\ref{fig:LPD_vs_k} illustrates how the minimal value $\inf_{\nu}\Pr_{\nu} (T \leq \nu + k \mid T > \nu, N = k)$ depends on $k$ for all algorithms when $\Durations = \smallset{5, 6, \cdots, 10}$ and $\LCPFA_m \approx 0.05$.
For CUSUM, WL CUSUM, and FMA the minimal value of the probability $\Pr_{\nu} (T \leq \nu + k \mid T > \nu, N = k)$ is attained at $\nu = 0$ while for mFMA it is attained at $\nu = 2$.
The WL CUSUM procedure performs better than CUSUM.
The mFMA performs better than the classical FMA.
Not surprisingly, FMA and mFMA (which are by design tuned to $\inf \Durations$) perform significantly better than its competitors at low values of the change duration, while at larger values of the change duration WL CUSUM and CUSUM take the lead.
\begin{figure}
    \centering
    \includegraphics[width=\textwidth]{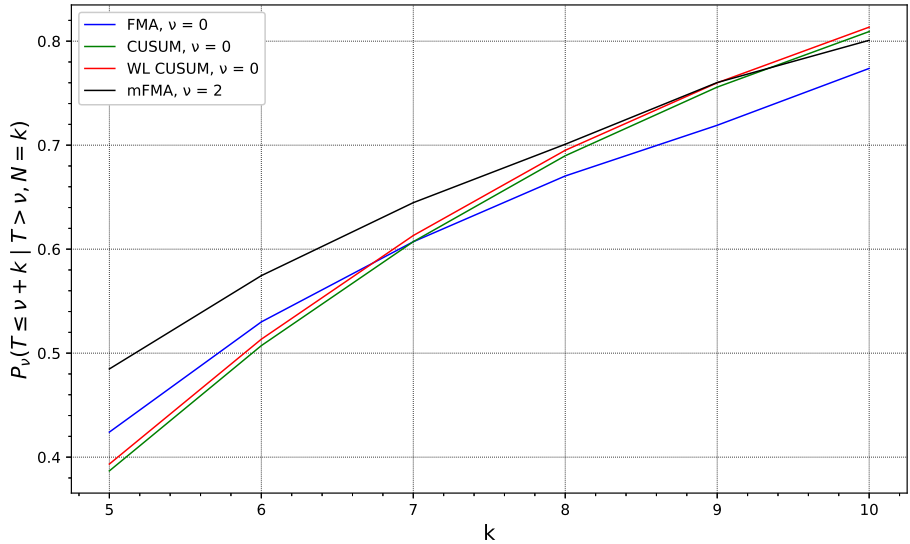}
    \caption{Comparison of $\Pr_{\nu} (T \leq \nu+k \mid T>\nu, N = k)$ vs $k$ for FMA, mFMA, CUSUM, WL CUSUM with $\LCPFA_m \approx 0.05$ for the case $\Durations = \smallset{5, 6, \cdots, 10}$. For each procedure $\nu$ is chosen where the infimum in \eqref{eq:lpd:conditional} is reached.}
    \label{fig:LPD_vs_k}
\end{figure}

The comparison of all rules showing $\LPD$ as a function of $\LCPFA$ for the case $\Durations = \smallset{5, 6, \cdots, 10}$ is presented in Table~\ref{t:LPD_vs_LCPFA} and Figure~\ref{fig:LPD_vs_LCPFA_log_scale_5_10}.
Here $\SE (T_\WLCUSUM)$, $\SE (T_\FMA)$, and $\SE (T_{\MFMA})$ are standard errors when evaluating $\LPD_\pi$ for the three procedures.

\begin{table}
    \newcommand{\ten}[2]{$#1 \cdot\! 10^{#2}$}
    \tbl{Operating characteristics of the WL CUSUM, CUSUM and FMA algorithms for $\Durations = \smallset{5, 6, \cdots, 10}$ .}
    {\begin{tabular}{@{} l c c c c c c c @{}} \toprule
        $\LCPFA_{m}$ & $10^{-1}$  & \ten{5}{-2} & \ten{2}{-2} & $10^{-2}$ & \ten{5}{-3} & $10^{-3}$ & $10^{-4}$ \\
        \midrule
        $\LPD_\pi(T_\WLCUSUM)$         & 0.7444 & 0.6350 & 0.4970 & 0.3950 & 0.3139 & 0.1730 & 0.0639 \\
        $\SE (T_\WLCUSUM)$             & 0.0013 & 0.0016 & 0.0017 & 0.0016 & 0.0015 & 0.0010 & 0.0005 \\
        \midrule
        $\LPD_\pi^{ie}(T_\CUSUM)$      & 0.7415 & 0.6326 & 0.4769 & 0.3655 & 0.2794 & 0.1290 & 0.0305 \\

        \midrule
        $\LPD_\pi(T_\FMA)$             & 0.7291 & 0.6214 & 0.4719 & 0.3841 & 0.2977 & 0.1558 & 0.0514 \\
        $\SE (T_\FMA)$                 & 0.0014 & 0.0016 & 0.0018 & 0.0017 & 0.0014 & 0.0009 & 0.0003 \\
        \midrule
        $\LPD_\pi(T_{\MFMA)}$          & 0.7672 & 0.6631 & 0.5126 & 0.4181 & 0.3258 & 0.1666 & 0.0556 \\
        $\SE (T_{\MFMA})$              & 0.0012 & 0.0016 & 0.0018 & 0.0018 & 0.0015 & 0.0009 & 0.003  \\ \bottomrule
    \end{tabular}}
    \label{t:LPD_vs_LCPFA}
\end{table}
\begin{figure}
    \centering
    \includegraphics[width=\textwidth]{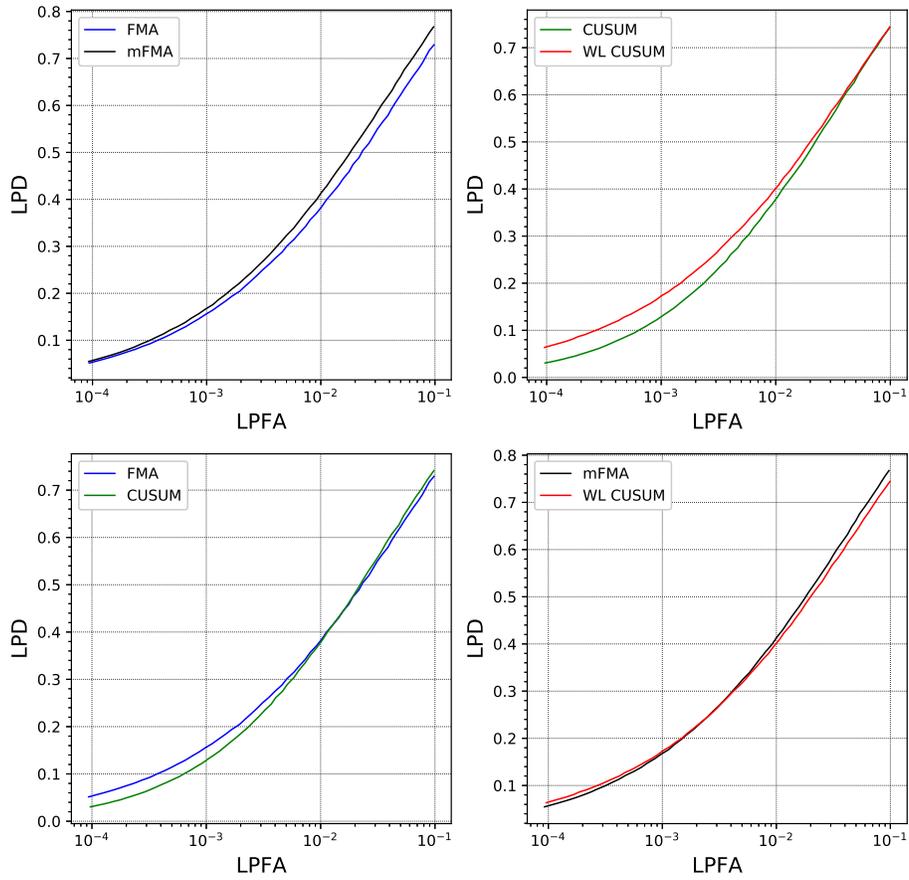}
    \caption{Comparison of operating characteristics ($\LPD$ vs $\LCPFA$) of the FMA, mFMA, CUSUM and WL CUSUM procedures for $\Durations = \smallset{5, 6, \cdots, 10}$; horizontal axis log-scale.}
    \label{fig:LPD_vs_LCPFA_log_scale_5_10}
\end{figure}

Figure~\ref{fig:LPD_vs_LCPFA_log_scale_5_10} shows that the window-limited CUSUM procedure performs better than the classic CUSUM procedure.
Moreover, the difference increases when $\LCPFA \to 0$. It also shows that the mFMA procedure performs better than the classic FMA.
However, as expected, the difference decreases when $\LCPFA \to 0$. The mFMA procedure performs much better than its competitors at not quite small $\LCPFA$ values, while the window-limited CUSUM procedure performs significantly better than its competitors at $\LCPFA \to 0$.
However, when $\LCPFA$ tends to zero, the values of $\LPD$ are rather poor for practical purposes.
More often in practice, the duration of the change is longer, as in the second case where $N \in \Durations = \smallset{7, 8, \cdots, 15}$.

The comparison of all rules showing $\LPD$ as a function of $\LCPFA$ for $\Durations = \smallset{7, 8, \cdots, 15}$ is presented in Table~\ref{t:LPD_vs_LCPFA_7_15} and Figure~\ref{fig:LPD_vs_LCPFA_log_7_15_scale}.
\begin{table}
    \newcommand{\ten}[2]{$#1 \cdot\! 10^{#2}$}
    \tbl{Operating characteristics of the WL CUSUM, CUSUM and FMA algorithms for $\Durations = \smallset{7, 8, \cdots, 15}$.}
    {\begin{tabular}{@{} l c c c c c c c @{}} \toprule
        $\LCPFA_{m}$ & $10^{-1}$  & \ten{5}{-2} & \ten{2}{-2} & $10^{-2}$ & \ten{5}{-3} & $10^{-3}$ & $10^{-4}$ \\
        \midrule
        $\LPD_\pi(T_\WLCUSUM)$         & 0.8549 & 0.7829 & 0.6770 & 0.5842 & 0.5129 & 0.3250 & 0.1629 \\
        $\SE (T_\WLCUSUM)$             & 0.0007 & 0.0010 & 0.0012 & 0.0013 & 0.0013 & 0.0012 & 0.0008 \\
        \midrule
        $\LPD_\pi^{ie}(T_\CUSUM)$      & 0.8551 & 0.7812 & 0.6676 & 0.5738 & 0.4953 & 0.3167 & 0.1370 \\

        \midrule
        $\LPD_\pi(T_\FMA)$             & 0.8514 & 0.7680 & 0.6528 & 0.5552 & 0.4716 & 0.2824 & 0.1205 \\
        $\SE (T_\FMA)$                 & 0.0007 & 0.0010 & 0.0013 & 0.0014 & 0.0014 & 0.0011 & 0.0006 \\
        \midrule
        $\LPD_\pi(T_{\MFMA})$          & 0.8734 & 0.7945 & 0.6797 & 0.5813 & 0.4947 & 0.2962 & 0.1262 \\
        $\SE (T_{\MFMA})$              & 0.0007 & 0.0009 & 0.0012 & 0.0014 & 0.0015 & 0.0012 & 0.0046 \\ \bottomrule
    \end{tabular}}
    \label{t:LPD_vs_LCPFA_7_15}
\end{table}
\begin{figure}[!h!]
    \centering
    \includegraphics[width=\textwidth]{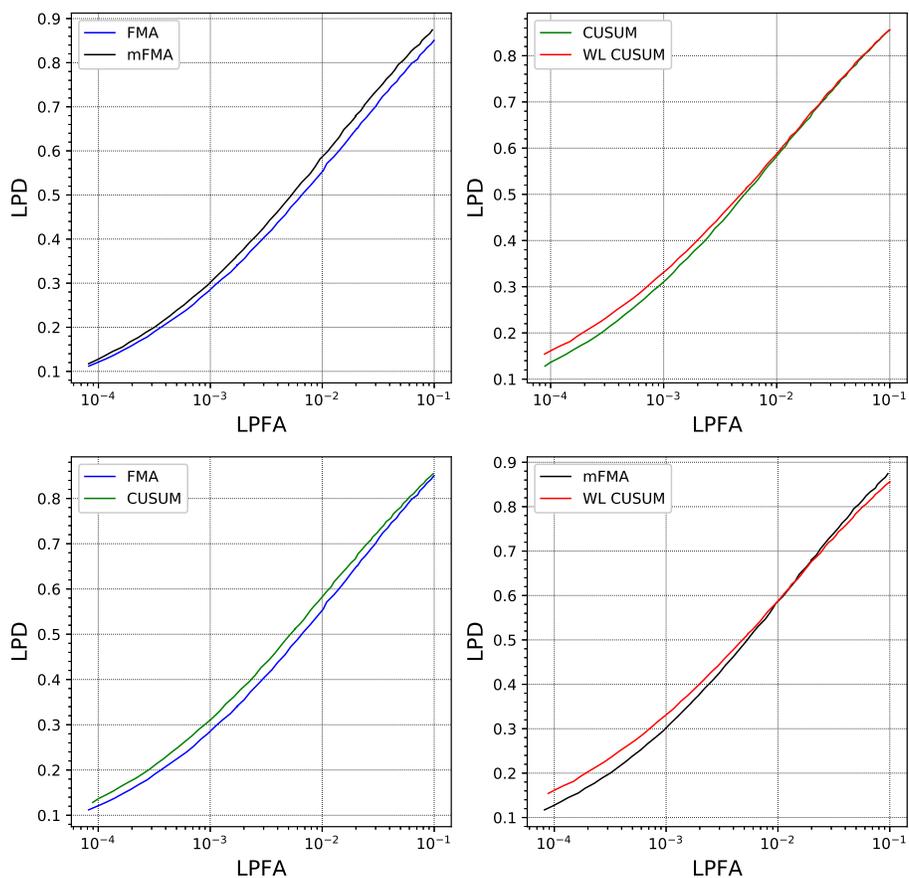}
   \caption{Comparison of operating characteristics ($\LPD$ vs $\LCPFA$) of the FMA, mFMA, CUSUM and WL CUSUM for $\Durations = \smallset{7, 8, \cdots, 15}$; horizontal axis log-scale.}
    \label{fig:LPD_vs_LCPFA_log_7_15_scale}
\end{figure}

The conclusions that were made in the first case hold for other values of the anomaly duration. However, for more realistic values of the change duration, the $\LPD$ values are noticeably higher.

To corroborate the recommended choice of window size (see Section \ref{sec:candidates}) for window-limited CUSUM (maximum of all possible change duration values, i.e., $M = \sup \Durations$) and for mFMA (smallest possible change duration values, i.e., $M = \inf \Durations$), we performed a numerical comparison of the performance ($\LPD$ vs $\LCPFA$) of the considered rules tuned to different window size $M$.
Figure~\ref{fig:LPD_vs_LPFA_for_dif_wl_mFMA_and_WL_CUS} shows that the window-limited CUSUM procedure, for which the window size is equal to the maximum of all possible change duration values, i.e., $M = 10$, performs better than the window-limited CUSUM procedure, for which the window size is $M = 5$.
In contrast, the mFMA procedure, for which the window size is equal to the smallest possible change duration value, i.e., $M = 5$, performs better than the mFMA procedure, for which the window size is equal to $M = 10$.

\begin{figure}
    \centering
    \includegraphics[width=\textwidth]{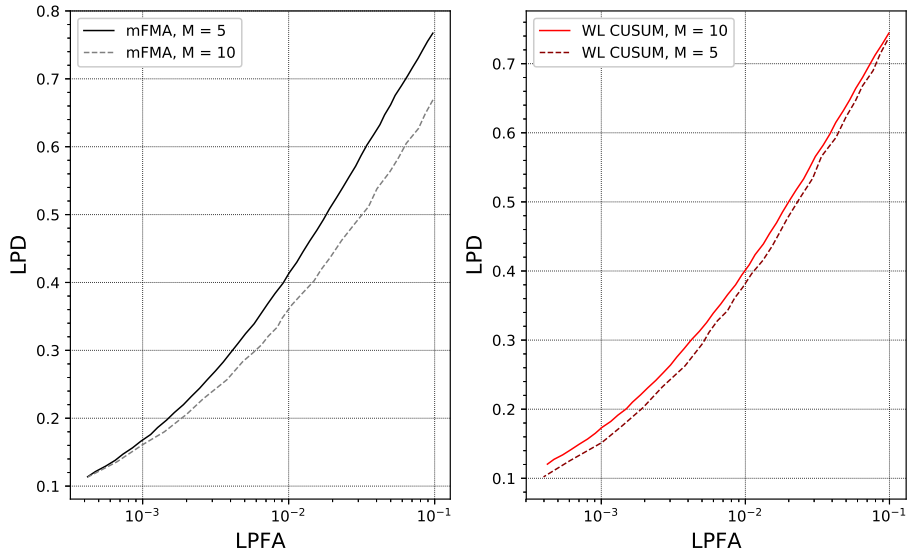}
    \caption{Performance comparison ($\LPD$ vs $\LCPFA$) for mFMA and WL CUSUM tuned to different window sizes $M$ for the case where change duration varies $\Durations = \smallset{5, 6, \cdots, 10}$.}
    \label{fig:LPD_vs_LPFA_for_dif_wl_mFMA_and_WL_CUS}
\end{figure}

The results of the comparison operating characteristics of WL CUSUM and mFMA obtained by using Lemma \ref{lemma_3} and Lemma \ref{lemma_4} (theoretical bounds for the operating characteristics) and by Monte Carlo simulations are presented in Table~\ref{t:WL_CUS_tb_vs_WL_CUS_MC} and Table~\ref{t:mFMA_tb_vs_mFMA_MC}, respectively. As expected, the upper bounds are very conservative.

\begin{table}[h]
    \newcommand{\ten}[2]{$#1 \cdot\! 10^{#2}$}
    \tbl{Operating characteristics for WL CUSUM (theoretical bounds vs MC simulations).}
    {\begin{tabular}{@{} l c c c c c @{}} \toprule
        \multicolumn{6}{c}{Upper bound for LPFA} \\
        \midrule
        theoretical bounds        & 0.4724 & 0.2507 & 0.0939 & 0.0413 & 0.0174 \\
        MC simulations            & 0.0999 & 0.0497 & 0.0195 & 0.0096 & 0.0049 \\
        \midrule
        \multicolumn{6}{c}{Lower bound for LPD} \\
        \midrule
        theoretical bounds        & 0.612 & 0.521 & 0.403 & 0.320 & 0.246 \\
        MC simulations            & 0.744 & 0.635 & 0.490 & 0.389 & 0.304 \\ \bottomrule
    \end{tabular}}
    \label{t:WL_CUS_tb_vs_WL_CUS_MC}
\end{table}
\begin{table}[!h!]
    \newcommand{\ten}[2]{$#1 \cdot\! 10^{#2}$}
    \tbl{Operating characteristics for mFMA (theoretical bounds vs MC simulations).}
    {\begin{tabular}{@{} l c c c c c @{}} \toprule
        \multicolumn{6}{c}{Upper bound for LPFA} \\
        \midrule
        theoretical bounds        & 0.1617 & 0.0806 & 0.0294 & 0.0136 & 0.0067 \\
        MC simulations            & 0.0985 & 0.0493 & 0.0191 & 0.0097 & 0.0049 \\
        \midrule
        \multicolumn{6}{c}{Lower bound for LPD} \\
        \midrule
        theoretical bounds        & 0.551 & 0.438 & 0.304 & 0.224 & 0.166 \\
        MC simulations            & 0.767 & 0.664 & 0.514 & 0.407 & 0.321 \\ \bottomrule
    \end{tabular}}
    \label{t:mFMA_tb_vs_mFMA_MC}
\end{table}

\subsection{Asymptotic exponentiality of stopping times} \label{subsec:sim:arl}

We now investigate whether the distribution of the stopping times of change detection procedures under $\Pr_\infty$ (i.e., under the no-change hypothesis) is close to exponential.
We use QQ plots to assess whether this is the case. Specifically, we plot the empirical quantiles of the observed stopping time against the theoretical quantiles of the geometric distribution with parameter $K/\sum_{i = 1}^K T_i$, where $T_i$'s are the generated stopping times.
To accomplish this we perform Monte Carlo simulations with $10^{7}$ runs for each of the four detection rules with thresholds chosen so that the $\ARL$ of each detection procedure is approximately equal to $200$.
Window-limited CUSUM and FMA are configured as in our first scenario, i.e., to detect a change of duration from $5$ to $10$. For each Monte Carlo run, we get the stopping time assuming the change never occurs.
The QQ plots in Figure~\ref{fig:QQ_plots} suggest that the distributions are indeed close to geometric even for moderate values of the ARL to false alarm.
\begin{figure}
    \centering
    \begin{subfigure}{0.49\columnwidth}
        \includegraphics[width=\textwidth]{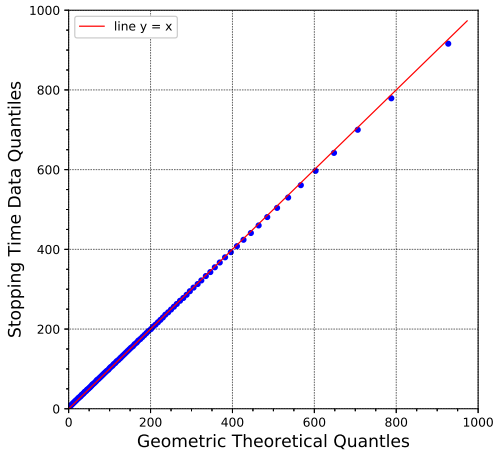}
        \caption{FMA}
    \end{subfigure}
    \begin{subfigure}{0.49\columnwidth}
        \includegraphics[width=\textwidth]{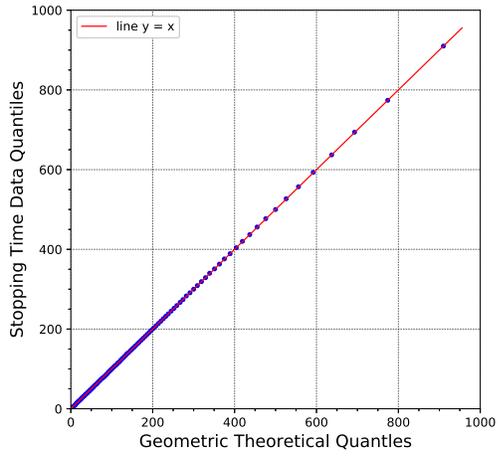}
        \caption{Modified FMA}
    \end{subfigure}
    \begin{subfigure}{0.49\columnwidth}
        \includegraphics[width=\textwidth]{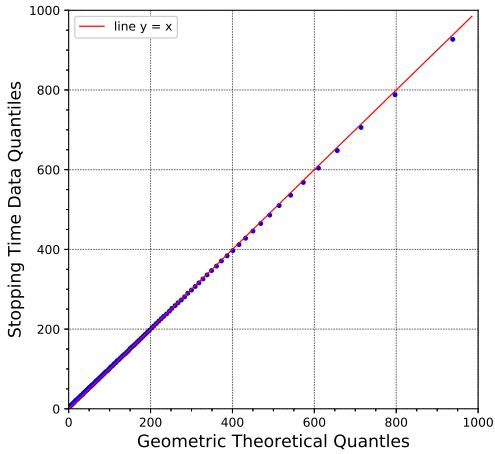}
        \caption{CUSUM}
    \end{subfigure}
    \begin{subfigure}{0.49\columnwidth}
        \includegraphics[width=\textwidth]{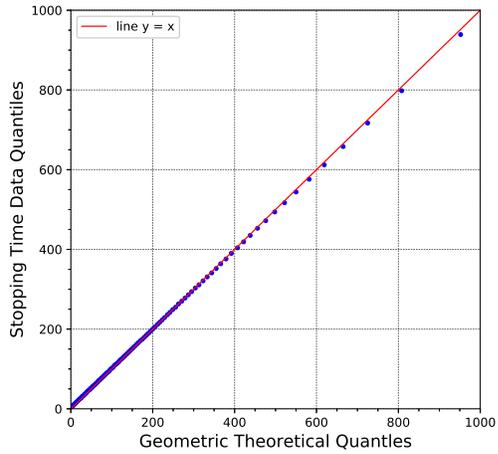}
        \caption{WL CUSUM}
    \end{subfigure}

    \caption{Quantile-Quantile (QQ) plots for FMA, Modified FMA, CUSUM, WL CUSUM when $\ARL \approx 200$. The $x$-axis shows the theoretical quantiles of the geometric distribution with parameter $\approx 1/200$ and the $y$-axis shows the quantiles of distributions of the observed stopping times for each of the four detection rules.}
    \label{fig:QQ_plots}
\end{figure}
Hence, the following approximation to $\LCPFA_{m}$ may be used:
\begin{align} \label{eq:LCPFA_ARL}
    \LCPFA_{exp}(T) = 1 - \Bigg(1 - \frac{1}{\ARL(T)}\Bigg)^{m}.
\end{align}
This allows one to estimate $\LCPFA$ by using $\ARL$, which significantly reduces the computational burden associated with MC simulations in Section~\ref{subsec:sim:performance}.
Comparison of $\LCPFA_{exp}$ obtained using approximation \eqref{eq:LCPFA_ARL} and simulated $\ARL$ with MC estimate $\LCPFA_{mc}$, as described in Section~\ref{subsec:sim:performance}, is presented in Table~\ref{t:LCPFA_mc_vs_LCPFA_th}.
As can be seen, the approximation \eqref{eq:LCPFA_ARL} is very accurate for all change detection rules. Thus, this approximation is useful in most practical problems, simplifying the selection of thresholds.

\begin{table}
    \newcommand{\ten}[2]{$#1 \cdot\! 10^{#2}$}
    \tbl{MC simulations vs asymptotic approximation for $\LCPFA_{m}$.}
    {\begin{tabular}{@{} l c c c c @{}} \toprule
        $\LCPFA_{mc}(T_\WLCUSUM)$   & \ten{9.91}{-2}     & \ten{9.45}{-3}     & \ten{10.1}{-4}     & \ten{9.6}{-5}  \\
        $\LCPFA_{exp}(T_\WLCUSUM)$  & \ten{9.31}{-2}     & \ten{9.31}{-3}     & \ten{10.0}{-4}     & \ten{9.9}{-5}  \\
        Deviation, $\%$             & 6.0                & 1.5                & 0.9                & 3.3            \\
        \midrule
        $\LCPFA_{ie}(T_\CUSUM)$     & \ten{9.79}{-2}     & \ten{10.4}{-3}     & \ten{10.2}{-4}     & \ten{9.7}{-5}  \\
        $\LCPFA_{exp}(T_\CUSUM)$    & \ten{9.45}{-2}     & \ten{10.1}{-3}     & \ten{10.3}{-4}     & \ten{10.8}{-5} \\
        Deviation, $\%$             & 3.5                & 2.3                & 1.0                & 11.8           \\
        \midrule
        $\LCPFA_{mc}(T_\FMA)$       & \ten{9.81}{-2}     & \ten{10.2}{-3}     & \ten{9.70}{-4}     & \ten{9.3}{-5}  \\
        $\LCPFA_{exp}(T_\FMA)$      & \ten{8.78}{-2}     & \ten{9.81}{-3}     & \ten{9.70}{-4}     & \ten{9.7}{-5}  \\
        Deviation, $\%$             & 10.5               & 4.0                & 0.1                & 3.7            \\
        \midrule
        $\LCPFA_{mc}(T_{\MFMA})$    & \ten{9.72}{-2}     & \ten{10.1}{-3}     & \ten{9.60}{-4}     & \ten{9.3}{-5}  \\
        $\LCPFA_{exp}(T_{\MFMA})$   & \ten{9.33}{-2}     & \ten{10.1}{-3}     & \ten{9.70}{-4}     & \ten{9.2}{-5}  \\
        Deviation, $\%$             & 4.1                & 0.2                & 0.9                & 0.7            \\
        \bottomrule
    \end{tabular}}
    \label{t:LCPFA_mc_vs_LCPFA_th}
\end{table}

Next, we study accuracy of two known approximations for $\ARL$ of the classical FMA -- specifically, \citeauthor{lai-as74}'s asymptotic approximation (\cite{lai-as74}) and \citeauthor{Noonan+Zhigljavsky:2020}'s approximation (\cite{Noonan+Zhigljavsky:2020}).

Noonan and Zhigljavsky studied a rule they dubbed the MOSUM test.
In the case of Gaussian observations, this procedure is equivalent to the classical FMA test.
Otherwise, MOSUM compares the observed signal, rather than the log-likelihood ratio process, to a threshold.
Their approximation to the $\ARL$ of $T_\FMA$ is as follows:
\begin{align} \label{eq:ARL_FMA_N_ZH}
    \ARL_{N-Zh}(T_\FMA) = -\frac{M \cdot F(2, h, h_{M})}{\theta_{M}(h)^{2} \log(\theta_{M}(h))} + M,
\end{align}
where $h = b + M/2$, $h_{M} = h + 0.8239/\sqrt{M}$, $b$ is the threshold in \eqref{eq:fma:time}, and
\begin{align*}
    F(1, h, h_{M}) = \Phi(h) \Phi(h_{M}) - \varphi(h_{M}) \left[h \Phi(h) + \varphi(h) \right].
\end{align*}
Here $\Phi$ is the standard normal cdf, $\varphi$ is the standard normal pdf, and
\begin{align*}
    F(2, h, h_{M}) &= \frac{\varphi^{2}(h_{M})}{2} \left[(h^{2} - 1 + \sqrt{\pi} h) \Phi(h) + (h + \sqrt{\pi}) \varphi(h) \right] - \\
        &{} - \varphi(h_{M}) \Phi(h_{M}) \left[(h + h_{M}) \Phi(h) + \varphi(h) \right] + \Phi(h) \Phi^{2}(h_{M}) + \\
        &{} + \int_{0}^{\infty} \Phi(h - x) \left[\varphi(h_{M} + x) \Phi(h_{M} - x) - \sqrt{\pi}\varphi^{2}(h_{M})\Phi(\sqrt{2}x) \right] dx,
\end{align*}
with $\theta_{M}(h) = \slfrac{F(2, h, h_{M})}{F(1, h, h_{M})}$.

Lai's approximation for $\ARL$ of $T_\FMA$ has a much simpler expression:
\begin{align} \label{eq:ARL_FMA_Lai}
    \ARL_{Lai}(T_\FMA) = (1 - \Phi((b + M/2)/\sqrt{M}))^{-1}.
\end{align}
Despite its simplicity, it is asymptotically exact, in that $\ARL_{Lai}(T_\FMA) \to \EV_\infty(T_\FMA)$ as $b \to \infty$.

For both approximations, we used the scenario that was used in the first case when the change duration varies from $5$ to $10$, i.e., $N \in \Durations = \smallset{5, 6, \cdots, 10}$.
We run a Monte Carlo simulation (with $10^{6}$ runs) for various threshold values and look at how simulated results align with the predictions. The results are presented in Table~\ref{t:ARL_mc_vs_ARL_N_Zh_vs_ARL_Lai}.
\begin{table}
    \newcommand{\ten}[2]{$#1 \cdot\! 10^{#2}$}
    \tbl{MC simulations of $\ARL(T_\FMA)$ vs.\ its approximation $\ARL_{N-Zh}$ and $\ARL_{Lai}$.}
    {\begin{tabular}{@{} l c c c c c c c @{}} \toprule
        Threshold              & 2.25     & 2.89      & 3.70      & 4.18       & 4.67     & 5.71    & 7.00   \\
        \midrule
        $\ARL(T_\FMA)$         & 109.63   & 211.47    & 545.50    & 1026.43    & 2032.5   & 10488   & 108960 \\
        \midrule
        $\ARL_{N-Zh}(T_\FMA)$  & 114.11   & 217.36    & 555.88    & 1047.30    & 2077.6   & 10902   & 115490 \\
        Deviation, $\%$        & 4.08     & 2.78      & 1.90      & 2.03       & 2.22     & 3.95    & 5.99   \\
        \midrule
        $\ARL_{Lai}(T_\FMA)$   & 59.44    & 126.17    & 359.36    & 713.09     & 1477.7   & 8325.4  & 92946  \\
        Deviation, $\%$        & 45.78    & 40.34     & 34.12     & 30.53      & 27.29    & 20.62   & 14.70  \\
        \bottomrule
    \end{tabular}}
    \label{t:ARL_mc_vs_ARL_N_Zh_vs_ARL_Lai}
\end{table}
It is clear that Lai's approximation gives poor accuracy for small threshold values. As the threshold increases, the deviation of the estimate by Lai's approximation from the estimate by Monte Carlo simulations decreases, as expected.
Noonan and Zhigljavsky's approximation gives good results: deviation from Monte Carlo simulations does not exceed $6\%$. However, for large thresholds, the deviation of the estimate by Noonan and Zhigljavsky's approximation from the estimate by Monte Carlo simulations increases. From this, we can conclude that, unlike Lai's approximation, Noonan and Zhigljavsky's approximation is not asymptotic.
Thus, using \eqref{eq:LCPFA_ARL} and \eqref{eq:ARL_FMA_N_ZH} may be recommended to a practitioner working in the low $\ARL$ setting, while Lai's approximation \eqref{eq:ARL_FMA_Lai} may be used as a conservative asymptotically exact lower bound for the $\ARL$ of FMA regardless of threshold.

\section{Conclusions and discussion}\label{sec:conclusions}

We provided a review of existing change detection frameworks and performance measures and examined their relevance for detecting intermittent changes.
A thorough discussion as to the differences between these settings prompted us to propose one specific choice of false alarm and correct detection measures.
The formulation handles changes of unknown duration, whether the length of change is deterministic or random.
Furthermore, the transient nature of the signal coupled with the maximum likelihood principle yielded three detection rules that are equivalent to three stopping times popular in literature: CUSUM, window-limited CUSUM, and FMA.
In addition, a particular property when maximizing the likelihood ratio for FMA brought to light a modification of the FMA rule that, to the best of our knowledge, has not been considered before.
The simulation study further supported the conjecture that the modified FMA is superior to the classical one.

We presented ways to design each of the detection rules and performed a comparative numerical analysis between the four.
The window-limited CUSUM procedure shows operating characteristics better than the CUSUM procedure, and the modified FMA procedure performs better than the classic FMA.
For low values of the change duration FMA and modified FMA perform significantly better than their competitors.
The modified FMA performs significantly better than the others.
For the CUSUM procedure, very accurate performance estimates can be obtained by solving integral equations.
For the window-limited CUSUM and the modified FMA, we obtained theoretical bounds on both $\LCPFA$ and $\LPD$.
The former not only allows one to control the false alarm rate and thus choose a threshold that guarantees that the rate does not exceed a prescribed value, but also turned out to be reasonably accurate.
The bound on $\LPD$, however, is typically rough and not recommended for practical purposes.

An important direction of further research would involve the case of deterministic signal duration $N$ and the choice of $\pi$ in \eqref{eq:lpd:conditional}.
One possible motivation could be as follows.
Consider a family of oracle rules $\smallcset{T_{N^\star}}{N^\star \geq 1}$ that are tuned to a particular change duration $N^\star$, i.e., solve
\[
    \sup _{T \in \class (m, \alpha )} \inf_{\nu \geq 0} \Pr_{\nu} (T \leq \nu + N^\star \mid T > \nu, N = N^\star).
\]
Then $\pi$ should be chosen so that compared to \emph{any} oracle rule $T_{N^\star}$, $N^\star \in \Durations$, the loss in performance of a rule satisfying \eqref{eq:optimal_time} should remain bounded as $\alpha \to 0$ with $m$ either fixed or $m = m(\alpha) \to \infty$ at a certain rate.

Another interesting avenue of research would be to investigate the connection between formulation \eqref{eq:lpd:lorden}, \eqref{eq:optimal_time} and formulation \eqref{eq:lpd:conditional}, \eqref{eq:optimal_time} in the asymptotic setting.

%


\begin{thebibliography}{59}
\newcommand{\enquote}[1]{``#1''}
\providecommand{\natexlab}[1]{#1}
\providecommand{\url}[1]{\normalfont{#1}}
\providecommand{\urlprefix}{}

\bibitem[Bakhache and Nikiforov(2000)]{Bakhache+Nikiforov:2000}
Bakhache, Bacem, and Igor Nikiforov. 2000. ``Reliable Detection of Faults in
  Measurement Systems.'' \emph{International Journal of Adaptive Control and
  Signal Processing} 14: 683 -- 700.

\bibitem[Bar-Shalom and Li(1993)]{BarshalomLi93}
Bar-Shalom, Y., and X.~R. Li. 1993. \emph{Estimation and Tracking: Principles,
  Techniques and Software}. Artech House Radar Library. Boston-London: Artech
  House.

\bibitem[Berenkov, Tartakovsky, and Kolessa(2020)]{BerenkovTarKol_EnT2020}
Berenkov, N.~R., A.~G. Tartakovsky, and A.~E. Kolessa. 2020. ``Reliable
  Detection of Dynamic Anomalies With Application to Extracting Faint Space
  Object Streaks From Digital Frames.'' In \emph{2020 International Conference
  on Engineering and Telecommunication (EnT-MIPT 2020)}, Dolgoprudny, Russia,
  25-26 November.

\bibitem[Blackman, Dempster, and Broida(1993)]{Black1}
Blackman, S.~S., R.~J. Dempster, and T.~J. Broida. 1993. ``Multiple Hypothesis
  Track Confirmation for Infrared Surveillance Systems.'' \emph{IEEE
  Transactions on Aerospace and Electronic Systems} 29 (3): 810--823.

\bibitem[Broder and Schwartz(1990)]{Broder+Schwartz:1990}
Broder, Bruce, and Stuart~C. Schwartz. 1990. ``Quickest Detection Procedures
  and Transient Signal Detection.'' .

\bibitem[Debar, Dacier, and Wespi(1999)]{Debaretal-CN99}
Debar, H., M.~Dacier, and A.~Wespi. 1999. ``Toward a Taxonomy of Intrusion
  Detection Systems.'' \emph{Computer Networks} 31 (8): 805--822.

\bibitem[Duncan(1986)]{duncan-book86}
Duncan, Acheson~Johnston. 1986. \emph{Quality Control and Industrial Statistics
  (5th ed.)}. Richard D. Irwin Professional Publishing, Inc.

\bibitem[Ebrahimzadeh and Tchamkerten(2015)]{Ebrahimzadeh:2015}
Ebrahimzadeh, Ehsan, and Aslan Tchamkerten. 2015. ``Sequential Detection of
  Transient Changes in Stochastic Systems Under a Sampling Constraint.''
  \emph{2015 IEEE International Symposium on Information Theory (ISIT)}
  156--160.

\bibitem[Egea-Roca et~al.(2018)]{Egea-Roca+et+al:2018}
Egea-Roca, Daniel, Jos\'e~A. L\'opez-Salcedo, Gonzalo Seco-Granados, and
  H.~Vincent~Poor. 2018. ``Performance Bounds for Finite Moving Average Tests
  in Transient Change Detection.'' \emph{IEEE Transactions on Signal
  Processing} 66 (6): 1594--1606.

\bibitem[Ellis and Speed(2001)]{EllisSpeed-Book01}
Ellis, Juanita, and Tim Speed. 2001. \emph{The Internet Security Guidebook:
  From Planning to Deployment}. Academic Press.

\bibitem[Esary, Proschan, and Walkup(1967)]{esary}
Esary, J.~D., F.~Proschan, and D.~W. Walkup. 1967. ``Association of Random
  Variables, With Applications.'' \emph{Annals of Mathematical Statistics} 38
  (5): 1466--1474.

\bibitem[Gu\'epi\'e, Fillatre, and Nikiforov(2012)]{Nikiforov+et+al:2012}
Gu\'epi\'e, Blaise~K\'evin, Lionel Fillatre, and Igor Nikiforov. 2012.
  ``Sequential Detection of Transient Changes.'' \emph{Sequential Analysis} 31
  (4): 528--547.

\bibitem[Gu\'epi\'e, Fillatre, and Nikiforov(2017)]{Nikiforov+et+al:2017}
Gu\'epi\'e, Blaise~K\'evin, Lionel Fillatre, and Igor Nikiforov. 2017.
  ``Detecting a Suddenly Arriving Dynamic Profile of Finite Duration.''
  \emph{IEEE Transactions on Information Theory} 63 (5): 3039--3052.

\bibitem[Jeske et~al.(2018{\natexlab{a}})]{Jeskeetal-ASMBI2018}
Jeske, Daniel~R., Nathaniel~T. Steven, Alexander~G. Tartakovsky, and James~D.
  Wilson. 2018{\natexlab{a}}. ``Statistical Methods for Network Surveillance.''
  \emph{Applied Stochastic Models in Business and Industry} 34 (4): 425--445.
  Discussion Paper.

\bibitem[Jeske et~al.(2018{\natexlab{b}})]{Jeskeetal-WileyRef2018}
Jeske, Daniel~R., Nathaniel~T. Steven, James~D. Wilson, and Alexander~G.
  Tartakovsky. 2018{\natexlab{b}}. ``Statistical Network Surveillance.''
  \emph{Wiley StatsRef: Statistics Reference Online} 1--12.

\bibitem[Kent(2000)]{Kent}
Kent, S. 2000. ``On the Trail of Intrusions Into Information Systems.''
  \emph{IEEE Spectrum} 37 (12): 52--56.

\bibitem[Lai(1974)]{lai-as74}
Lai, Tze~Leung. 1974. ``Control Charts Based on Weighted Sums.'' \emph{Annals
  of Statistics} 2 (1): 134--147.

\bibitem[Lai(1998)]{LaiIEEE98}
Lai, Tze~Leung. 1998. ``Information Bounds and Quick Detection of Parameter
  Changes in Stochastic Systems.'' \emph{IEEE Transactions on Information
  Theory} 44 (7): 2917--2929.

\bibitem[Liang, Tartakovsky, and
  Veeravalli(2022)]{LiangTartakovskyVeerIEEEIT2022}
Liang, Yuchen, Alexander~G. Tartakovsky, and Venugopal~V. Veeravalli. 2022.
  ``Quickest Change Detection With Non-Stationary Post-Change Observations.''
  \emph{IEEE Transactions on Information Theory} 68: Early Access.

\bibitem[Lorden(1971)]{lorden-ams71}
Lorden, Gary. 1971. ``Procedures for Reacting to a Change in Distribution.''
  \emph{Annals of Mathematical Statistics} 42 (6): 1897--1908.

\bibitem[Mana et~al.(2022)]{NikiforovIFAC2022}
Mana, Fatima~Ezzahra, Blaise~K\'evin Gu\'epi\'e, Rapha\`ele Deprost, Eric
  Herber, and Igor Nikiforov. 2022. ``The Air Pollution Monitoring by
  Sequential Detection of Transient Changes.'' \emph{IFAC Papers online} 55-5:
  60--65.

\bibitem[Mana, Gu\'epi\'e, and Nikiforov(2023)]{Nikiforov+et+al:2023}
Mana, Fatima~Ezzahra, Blaise~K\'evin Gu\'epi\'e, and Igor Nikiforov. 2023.
  ``Sequential Detection of an Arbitrary Transient Change Profile by the FMA
  Test.'' \emph{Sequential Analysis} 1--21.

\bibitem[Mei(2008)]{Mei-SQA08}
Mei, Yajun. 2008. ``Is Average Run Length to False Alarm Always an Informative
  Criterion?'' \emph{Sequential Analysis} 27 (4): 354--376.

\bibitem[Moustakides(1986)]{MoustakidesAS86}
Moustakides, George~V. 1986. ``Optimal Stopping Times for Detecting Changes in
  Distributions.'' \emph{Annals of Statistics} 14 (4): 1379--1387.

\bibitem[Moustakides(2014)]{Moustakides:2014}
Moustakides, George~V. 2014. ``Multiple Optimality Properties of the Shewhart
  Test.'' \emph{Sequential Analysis} 33 (3): 318--344.

\bibitem[Moustakides, Polunchenko, and Tartakovsky(2011)]{MoustPolTarSS09}
Moustakides, George~V., Aleksey~S. Polunchenko, and Alexander~G. Tartakovsky.
  2011. ``A Numerical Approach to Performance Analysis of Quickest Change-Point
  Detection Procedures.'' \emph{Statistica Sinica} 21 (2): 571--596.

\bibitem[Noonan and Zhigljavsky(2020)]{Noonan+Zhigljavsky:2020}
Noonan, Jack, and Anatoly Zhigljavsky. 2020. ``Power of the MOSUM Test for
  Online Detection of a Transient Change in Mean.'' \emph{Sequential Analysis}
  39 (2): 269--293.

\bibitem[Noonan and Zhigljavsky(2021)]{Noonan+Zhigljavsky:2021}
Noonan, Jack, and Anatoly Zhigljavsky. 2021. ``Approximations for the Boundary
  Crossing Probabilities of Moving Sums of Random Variables.''
  \emph{Methodology and Computing in Applied Probability} 23 (3): 873--892.

\bibitem[Ortner and Nehorai(2007)]{Ortner+Nehorai:2007}
Ortner, M., and A.~Nehorai. 2007. ``A Sequential Detector for Biochemical
  Release in Realistic Environments.'' \emph{IEEE Transactions on Signal
  Processing} 55 (8): 4173-4182.

\bibitem[Page(1954)]{page-bka54}
Page, E.~S. 1954. ``Continuous Inspection Schemes.'' \emph{Biometrika} 41
  (1--2): 100--114.

\bibitem[Peng, Leckie, and Ramamohanarao(2004)]{peng-lncs04}
Peng, Tao, Christopher Leckie, and Kotagiri Ramamohanarao. 2004. ``{Proactively
  Detecting Distributed Denial of Service Attacks Using Source IP Address
  Monitoring}.'' In \emph{Networking 2004},  edited by Nikolas Mitrou, Kimon
  Kontovasilis, George~N. Rouskas, Ilias Iliadis, and Lazaros Merakos, Vol.
  3042 of \emph{Lecture Notes in Computer Science}, 771--782. Berlin, DE:
  Springer-Verlag.

\bibitem[Pollak and Tartakovsky(2009)]{PollakTartakovskyTPA09}
Pollak, M., and A.~G. Tartakovsky. 2009. ``Asymptotic Exponentiality of the
  Distribution of First Exit Times for a Class of {Markov} Processes With
  Applications to Quickest Change Detection.'' \emph{Theory of Probability and
  its Applications} 53 (3): 430--442.

\bibitem[Pollak(1985)]{PollakAS85}
Pollak, Moshe. 1985. ``Optimal Detection of a Change in Distribution.''
  \emph{Annals of Statistics} 13 (1): 206--227.

\bibitem[Polunchenko and Tartakovsky(2010)]{PolunTartakovskyAS10}
Polunchenko, A.~S., and A.~G. Tartakovsky. 2010. ``On Optimality of the
  {Shiryaev--Roberts} Procedure for Detecting a Change in Distribution.''
  \emph{Annals of Statistics} 38 (6): 3445--3457.

\bibitem[Raghavan, Galstyan, and Tartakovsky(2013)]{Raghavanetal-AoAS2013}
Raghavan, V., A.~Galstyan, and A.~G. Tartakovsky. 2013. ``Hidden {M}arkov
  Models for the Activity Profile of Terrorist Groups.'' \emph{Annals of
  Applied Statistics} 7: 2402--24307.

\bibitem[Robbins(1954)]{Robbins1954}
Robbins, H.~E. 1954. ``A Remark on the Joint Distribution of Cumulative Sums.''
  \emph{Annals of Mathematical Statistics} 25: 614--616.

\bibitem[Rovatsos, Zou, and Veeravalli(2017)]{Rovatsos+Zou+Veeravalli:2017}
Rovatsos, Georgios, Shaofeng Zou, and Venugopal~V. Veeravalli. 2017. ``Quickest
  Change Detection Under Transient Dynamics.'' In \emph{2017 IEEE International
  Conference on Acoustics, Speech, and Signal Processing, ICASSP 2017 -
  Proceedings}, ICASSP, IEEE International Conference on Acoustics, Speech and
  Signal Processing - Proceedings, Jun., 4785--4789. Institute of Electrical
  and Electronics Engineers Inc.

\bibitem[Siegmund(1985)]{siegmund-book85}
Siegmund, David. 1985. \emph{Sequential Analysis: Tests and Confidence
  Intervals}. Series in Statistics. New York, USA: Springer-Verlag.

\bibitem[Spivak and Tartakovsky(2020)]{SpivakTar_EnT2020}
Spivak, V.S., and A.~G. Tartakovsky. 2020. ``Efficient Algorithm for
  Initialization of Object Tracks Based on Changepoint Detection Method.'' In
  \emph{2020 International Conference on Engineering and Telecommunication
  (EnT-MIPT 2020)}, Dolgoprudny, Russia, 25-26 November.

\bibitem[Tartakovsky(2005)]{TartakovskyIEEECDC05}
Tartakovsky, A.~G. 2005. ``Asymptotic Performance of a Multichart {CUSUM} Test
  Under False Alarm Probability Constraint.'' In \emph{Proceedings of the 44th
  IEEE Conference Decision and Control and European Control Conference
  (CDC-ECC'05), Seville, SP}, 320--325. IEEE, Omnipress CD-ROM.

\bibitem[Tartakovsky(2014)]{Tartakovsky-Cybersecurity14}
Tartakovsky, A.~G. 2014. ``Rapid Detection of Attacks in Computer Networks by
  Quickest Changepoint Detection Methods.'' In \emph{Data Analysis for Network
  Cyber-Security},  edited by N.~Adams and N.~Heard, 33--70. London, UK:
  Imperial College Press.

\bibitem[Tartakovsky(2020)]{Tartakovsky_book2020}
Tartakovsky, A.~G. 2020. \emph{Sequential Change Detection and Hypothesis
  Testing: General Non-i.i.d. Stochastic Models and Asymptotically Optimal
  Rules}. Monographs on Statistics and Applied Probability 165. Boca Raton,
  London, New York: Chapman \& Hall/CRC Press, Taylor \& Francis Group.

\bibitem[Tartakovsky and Brown(2008)]{Tartakovsky&Brown-IEEEAES08}
Tartakovsky, A.~G., and J.~Brown. 2008. ``Adaptive Spatial-Temporal Filtering
  Methods for Clutter Removal and Target Tracking.'' \emph{IEEE Transactions on
  Aerospace and Electronic Systems} 44 (4): 1522--1537.

\bibitem[Tartakovsky, Nikiforov, and Basseville(2014)]{TNB_book2014}
Tartakovsky, A.~G., I.~V. Nikiforov, and M.~Basseville. 2014. \emph{Sequential
  Analysis: Hypothesis Testing and Changepoint Detection}. Monographs on
  Statistics and Applied Probability. Boca Raton, London, New York: Chapman \&
  Hall/CRC Press.

\bibitem[Tartakovsky et~al.(2021{\natexlab{a}})]{TartakovskyetalIEEESP2021}
Tartakovsky, A.G., N.R. Berenkov, A.E. Kolessa, and I.V. Nikiforov.
  2021{\natexlab{a}}. ``Optimal Sequential Detection of Signals With Unknown
  Appearance and Disappearance Points in Time.'' \emph{IEEE Transactions on
  Signal Processing} 69: 2653--2662.

\bibitem[Tartakovsky(2002)]{Tartakovsky-IEEEASC2002}
Tartakovsky, Alexander~G. 2002. ``An Efficient Adaptive Sequential Procedure
  for Detecting Targets.'' In \emph{Proceedings of the IEEE Aerospace
  Conference, Big Sky, MT, USA},  edited by David~A. Williamson, Vol.~4, Mar.,
  1581--1596. IEEE.

\bibitem[Tartakovsky(2008)]{Tartakovsky-SQA08a}
Tartakovsky, Alexander~G. 2008. ``Discussion on ``{Is} Average Run Length to
  False Alarm Always an Informative Criterion?'' by {Yajun Mei}.''
  \emph{Sequential Analysis} 27 (4): 396--405.

\bibitem[Tartakovsky et~al.(2021{\natexlab{b}})]{Tartakovsky+et+al:2021}
Tartakovsky, Alexander~G., Nikita~R. Berenkov, Alexei~E. Kolessa, and Igor~V.
  Nikiforov. 2021{\natexlab{b}}. ``Optimal Sequential Detection of Signals With
  Unknown Appearance and Disappearance Points in Time.'' \emph{IEEE
  Transactions on Signal Processing} 69: 2653--2662.

\bibitem[Tartakovsky, Pollak, and
  Polunchenko(2012)]{tartakovskypolpolunch-tpa11}
Tartakovsky, Alexander~G., Moshe Pollak, and Aleksey~S. Polunchenko. 2012.
  ``Third-order Asymptotic Optimality of the Generalized {Shiryaev--Roberts}
  Changepoint Detection Procedures.'' \emph{Theory of Probability and its
  Applications} 56 (3): 457--484.

\bibitem[Tartakovsky and Polunchenko(2008)]{TartakovskyPolunchenko-FUSION08}
Tartakovsky, Alexander~G., and Aleksey~S. Polunchenko. 2008. ``Quickest
  Changepoint Detection in Distributed Multisensor Systems Under Unknown
  Parameters.'' In \emph{Proceedings of the 11th IEEE International Conference
  on Information Fusion, Cologne, DE}, Jul.

\bibitem[Tartakovsky et~al.(2006{\natexlab{a}})]{Tartakovskyetal-SM06}
Tartakovsky, Alexander~G., Boris~L. Rozovskii, Rudolf~B. Bla\'{z}ek, and
  Hongjoong Kim. 2006{\natexlab{a}}. ``Detection of Intrusions in Information
  Systems by Sequential Change-point Methods.'' \emph{Statistical Methodology}
  3 (3): 252--293.

\bibitem[Tartakovsky et~al.(2006{\natexlab{b}})]{Tartakovskyetal-IEEESP06}
Tartakovsky, Alexander~G., Boris~L. Rozovskii, Rudolf~B. Bla\'{z}ek, and
  Hongjoong Kim. 2006{\natexlab{b}}. ``A Novel Approach to Detection of
  Intrusions in Computer Networks Via Adaptive Sequential and Batch-sequential
  Change-point Detection Methods.'' \emph{IEEE Transactions on Signal
  Processing} 54 (9): 3372--3382.

\bibitem[Wang and Willett(2005{\natexlab{a}})]{Wang+Willett:2005b}
Wang, Z.J., and P.~Willett. 2005{\natexlab{a}}. ``Detecting Transients of
  Unknown Length.'' In \emph{2005 IEEE Aerospace Conference}, 2236--2247.

\bibitem[Wang and Willett(2005{\natexlab{b}})]{Wang+Willett:2005a}
Wang, Z.J., and P.~Willett. 2005{\natexlab{b}}. ``A Variable Threshold Page
  Procedure for Detection of Transient Signals.'' \emph{IEEE Transactions on
  Signal Processing} 53 (11): 4397--4402.

\bibitem[Willsky and Jones(1976)]{willsky-ac76}
Willsky, Alan~S., and Harold~L. Jones. 1976. ``A Generalized Likelihood Ratio
  Approach to the Detection and Estimation of Jumps in Linear Systems.''
  \emph{IEEE Transactions on Automatic Control} 21 (1): 108--112.

\bibitem[Woodroofe(1982)]{woodroofe-book82}
Woodroofe, Michael. 1982. \emph{Nonlinear Renewal Theory in Sequential
  Analysis}. Vol.~39 of \emph{CBMS-NSF Regional Conference Series in Applied
  Mathematics}. Philadelphia, PA, USA: SIAM.

\bibitem[Yakir(1995)]{Yakir-AS95}
Yakir, B. 1995. ``A Note on the Run Length to False Alarm of a Change-point
  Detection Policy.'' \emph{Annals of Statistics} 23 (1): 272--281.

\bibitem[Zou, Fellouris, and
  Veeravalli(2017{\natexlab{a}})]{Zou+Fellouris+Veeravalli:SIT:2017}
Zou, Shaofeng, Georgios Fellouris, and Venugopal~V. Veeravalli.
  2017{\natexlab{a}}. ``Asymptotic Optimality of D-CuSum for Quickest Change
  Detection Under Transient Dynamics.'' In \emph{2017 IEEE International
  Symposium on Information Theory, ISIT 2017}, IEEE International Symposium on
  Information Theory - Proceedings, 08, 2263--2267. Institute of Electrical and
  Electronics Engineers Inc.

\bibitem[Zou, Fellouris, and
  Veeravalli(2017{\natexlab{b}})]{Zou+Fellouris+Veeravalli:TIT:2017}
Zou, Shaofeng, Georgios Fellouris, and Venugopal~V. Veeravalli.
  2017{\natexlab{b}}. ``Quickest Change Detection Under Transient Dynamics:
  Theory and Asymptotic Analysis.'' \emph{IEEE Transactions on Information
  Theory} PP.

\end{thebibliography}

\end{document}